%% file: BckCUSUM.tex
\documentclass[12pt]{article}
\usepackage{amsmath}
\usepackage{graphicx}
\usepackage{enumerate}
\usepackage{natbib}

\usepackage{amssymb, amsthm, tikz}
\usepackage[utf8]{inputenc}
\usepackage[colorlinks=true, linkcolor=blue, citecolor=blue]{hyperref}
\usetikzlibrary{decorations.pathreplacing}
\newtheorem{assumption}{Assumption}
\newtheorem{lemma}{Lemma}
\newtheorem{theorem}{Theorem}
\theoremstyle{remark}
\newtheorem{remark}{Remark}
\DeclareMathOperator*{\plim}{plim}
\DeclareMathOperator*{\argmax}{argmax}
\DeclareMathOperator*{\argmin}{argmin}
\DeclareMathOperator*{\argsup}{argsup}
\newcommand{\dd}{\,\mathrm{d}}
\newcommand{\Deq}{\overset{d}{=}}
\newcommand{\Dlim}{\overset{d}{\longrightarrow}}
\newcommand{\mT}{\lfloor m T \rfloor}
\newcommand{\sM}{\lfloor s M \rfloor}
\newcommand{\rT}{\lfloor r T \rfloor}
\newcommand{\sT}{\lfloor s T \rfloor}
\newcommand{\zT}{\lfloor z T \rfloor}

\newcommand{\pplim}{\overset{p}{\longrightarrow}}
\def\T{{\prime}}

\begin{document}

\def\spacingset#1{\renewcommand{\baselinestretch}%
{#1}\small\normalsize} \spacingset{1}


\title{\bf Backward CUSUM for Testing and Monitoring Structural Change
with an Application to COVID-19 Pandemic Data}

\author{Sven Otto \\
Institute of Finance and Statistics, University of Bonn, \\
sven.otto@uni-bonn.de \bigskip \\
Jörg Breitung\thanks{
    Corresponding author: Jörg Breitung, Institute of Econometrics and Statistics, University of Cologne, Albertus-Magnus-Platz, 50923 Cologne, Germany. Mail: breitung@statistik.uni-koeln.de}
 \\
Institute of Econometrics and Statistics, University of Cologne, \\ breitung@statistik.uni-koeln.de }

\maketitle

\begin{abstract}
\noindent
It is well known that the conventional cumulative sum (CUSUM) test suffers from low power and large detection delay.
In order to improve the power of the test, we propose two alternative statistics.
The backward CUSUM detector considers the recursive residuals in reverse chronological order, whereas the stacked backward CUSUM detector sequentially cumulates a triangular array of backwardly cumulated residuals.
A multivariate invariance principle for partial sums of recursive residuals is given, and the limiting distributions of the test statistics are derived under local alternatives.
In the retrospective context, the local power of the tests is shown to be substantially higher than that of the conventional CUSUM test if a break occurs in the middle or at the end of the sample.
When applied to monitoring schemes, the detection delay of the stacked backward CUSUM is found to be much shorter than that of the conventional monitoring CUSUM procedure.
Furthermore, we propose an estimator of the break date based on the backward CUSUM detector and show that in monitoring exercises this estimator tends to outperform the usual maximum likelihood estimator.
Finally, an application of the methodology to COVID-19 data is presented.
\end{abstract}

\vfill

\newpage
\spacingset{1.5}

\input{1_introduction}

\input{2_CUSUMprocess}

\input{3_CUSUMtests}

\input{4_localpower}

\input{5_infinitehorizon}

\input{6_breakpoint}

\input{7_simulations}

\input{8_covid19}

\input{9_conclusion}

\section*{Acknowledgements}

We are thankful to Holger Dette, Josua Gösmann, Alexander Mayer, Dominik Wied and three referees for their very helpful comments and suggestions which helped to improve the paper a lot.
Furthermore, the usage of the CHEOPS HPC cluster for parallel computing is gratefully acknowledged.

\section*{Supporting Information}

An accompanying \texttt{R}-package for all methods presented in this article is available online at \texttt{https://github.com/ottosven/backCUSUM}.

\appendix
\newpage
\section{Appendix: Technical proofs}

\input{proofs.tex}

\newpage
\bibliographystyle{apalike}
\bibliography{references}
\addcontentsline{toc}{section}{References}

\end{document}

%% file: 1_introduction.tex
\section{Introduction} \label{sec:introduction}

Cumulative sums have become a standard statistical tool for testing and monitoring structural changes in time series models.
The CUSUM test was introduced by \cite{brown1975} as a
test for structural breaks in the coefficients of a linear regression model
$y_t =  x_t^\T  \beta_t + u_t$ with time index $t$, where $\beta_t$ denotes the coefficient vector, $x_t$ is the vector of regressor variables and $u_t$ is a zero mean error term.
Under the null hypothesis, there is no structural change in $\beta_t$, while, under the alternative hypothesis, the coefficient vector changes at unknown time $T^* \leq T$.

Sequential tests, such as the CUSUM test, consist of a detector statistic and a critical boundary function.
The CUSUM detector sequentially cumulates standardized one-step ahead forecast errors, which are also referred to as recursive residuals.
The detector is evaluated for each time point within the testing period, and, if its path crosses the boundary function at least once, the null hypothesis is rejected. If the endpoint of the sample is fixed and the test is applied once to the full sample by comparing the path of the detector with the boundary function, the test is called a retrospective test (henceforth: {\it$R$-test}).
A variety of {\it$R$-tests} have been proposed in the literature (for recent reviews, see \citealp{robbins2011}, \citealp{aue2013}, and \citealp{casini2019c}).

Since the seminal work of \cite{chu1996}, increasing interest has been focused on monitoring structural stability in real time.
Sequential monitoring procedures (henceforth: {\it$M$-tests}) consist of a detector statistic and a boundary function that are evaluated for periods beyond some historical time span.
The monitoring time span with $t>T$ can either have a fixed endpoint $M < \infty$ or an infinite horizon.
In the fixed endpoint setting, the monitoring period starts at $T+1$ and ends at $M$, while the boundary function depends on the ratio $m=M/T$.
In case of an infinite horizon, the monitoring time span does not need to be specified before the monitoring procedure starts.
These two monitoring schemes are also referred to as closed-end and open-end procedures (see \citealp{kirch2015}).
The null hypothesis of no structural change is rejected whenever the path of the detector crosses some critical boundary function for the first time.
Monitoring procedures for a fixed endpoint were proposed in \cite{leisch2000}, \cite{zeileis2005}, \cite{wied2013}, and \cite{dette2020}, whereas \cite{chu1996}, \cite{horvath2004}, \cite{aue2006}, \cite{fremdt2015}, and \cite{gosmann2021} considered an infinite monitoring horizon. In recent years, {\it$M$-tests} have become popular as tests for speculative bubbles in financial markets (e.g. \citealp{phillips2011}, \citealp{homm2012}, \citealp{astill2018}).
\vspace{3ex}

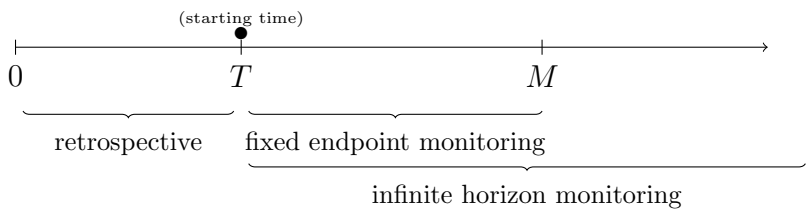
\begin{figure}[h]
\caption{Retrospective testing and monitoring}
\begin{center}
\begin{tikzpicture}
  \draw[->] (-0,0) -- (10,0) node [below] {};
  \draw (0,0.1) -- (0,-0.1) node [below] {$0$};
	\draw (3,0.1) -- (3,-0.1) node [below] {$T$};
	\draw (7,0.1) -- (7,-0.1) node [below] {$M$};
\tikzset{
    position label/.style={
       below = 3pt,
       text height = 1.5ex,
       text depth = 1ex
    },
   brace/.style={
     decoration={brace, mirror},
     decorate
   }
}
\node [position label] (here) at (3, 0.85) {\tiny{(starting time)}};
\node [position label] (here) at (3, 0.6) {$\bullet$};
\node [position label] (a) at (0.1,0) {};
\node [position label] (b1) at (2.9,0) {};
\node [position label] (b2) at (3.1,0) {};
\node [position label] (c) at (7,0) {};
\node [position label] (d) at (3.1,-0.7) {};
\node [position label] (e) at (10.5,-0.7) {};
\draw [brace] (a.south) -- (b1.south) node [position label, pos=0.5] {\footnotesize{retrospective}};	
\draw [brace] (b2.south) -- (c.south) node [position label, pos=0.5] {\footnotesize{fixed endpoint monitoring}};	
\draw [brace] (d.south) -- (e.south) node [position label, pos=0.5] {\footnotesize{infinite horizon monitoring}};		
\end{tikzpicture}
\end{center}
\label{testingschemes}
\end{figure}

A well-known drawback of the conventional CUSUM {\it$R$-test} is its low power, whereas the CUSUM {\it$M$-test} may exhibit large detection delays.
This is due to the fact that the pre-break recursive residuals are uninformative, as their expectation is equal to zero up to the break date, while the recursive residuals have a non-zero expectation after the break.
Hence, the cumulative sums of the recursive residuals contain a large number of uninformative residuals that only add noise to the statistic.
In contrast, if one cumulates the recursive residuals backwardly from the end of the sample to the beginning, the cumulative sum collects the informative residuals first, and the likelihood of exceeding the critical boundary will typically be larger than when cumulating residuals from the beginning onwards. In this paper, we show that backward CUSUM test procedures may indeed have a much higher power and lower detection delays than the conventional CUSUM {\it$R$-} and {\it$M$-tests}.
\citealt{phillips2018} proposed a reverse sample scheme for the PSY procedure that is used for detecting crises (bubble collapses). 
The main difference with our approach is that their regression is performed in reverse order, while our approach estimates the model in the original time but reverses the order of the (recursive) residuals.

Another way of motivating the backward CUSUM testing approach is to consider the simplest possible situation, where, under the null hypothesis, it is assumed that the process is generated as $y_t= \mu  + u_t$, with $\mu$ and $\sigma^2=Var(u_t)$ assumed to be known.
To test the hypothesis that the mean changes at $T^*$, we introduce the dummy variable $D_t^*$, which is unity for $t \geq T^*$ and zero elsewhere.
The uniformly most powerful test statistic is the $t$-statistic for the hypothesis $\delta=0$ in the regression $y_t-\mu = \delta D_t^*+ u_t$, which is given by $\sigma^{-1} (T-T^*+1)^{-1/2} \sum_{t=T^*}^T (y_t - \mu)$.
If $\mu$ is unknown, we may replace it by the full sample mean $\overline y$, resulting in the backward cumulative sum of the OLS residuals from period $T$ through $T^*$.
If $T^*$ is unknown, the test statistic is computed for all possible values of $T^*$, whereas the starting point $T$ of the backward cumulative sum remains constant.
Since the sum of the OLS residuals is zero, it follows that the test is equivalent to a test based on the forward cumulative sum of the OLS residuals.
In contrast, if we replace $\mu$ with the recursive mean $\overline \mu_{t-1}=(t-1)^{-1}\sum_{i=1}^{t-1} y_t$, we obtain a test statistic based on the backward cumulative sum of the recursive residuals (henceforth: backward CUSUM).
In this case, however, the test is different from a test based on the forward cumulative sum of the recursive residuals (henceforth: forward CUSUM).
This is due to the fact that the sum of the recursive residuals is an unrestricted random variable.
Accordingly, the two versions of the test may have quite different properties.
In particular, it turns out that the backward CUSUM is much more powerful than the standard forward CUSUM at the end of the sample.
Accordingly, this version of the CUSUM test procedure is better suited for the purpose of real-time monitoring, where it is crucial to be powerful at the end of the sample.

An additional problem of the conventional CUSUM test is that it has no power against alternatives that do not affect the unconditional mean of $y_t$ (see \citealp{kramer1988}).
For both retrospective testing and monitoring, we propose a multivariate sequential statistic in the fashion of the score-based cumulative sum statistic of \cite{hansen1992} and the tests by \cite{jiang2019}.
The maximum vector entry of the multivariate statistic yields a detector and a sequential test that has power against a much larger class of structural breaks than when using conventional CUSUM detectors.

We also suggest a new estimator for the break date based on backwardly cumulated recursive residuals.
This estimator outperforms the conventional estimator constructed by the sum of squared residuals whenever the break occurs close to the end of the sample, which is the relevant scenario for on-line monitoring.

This paper is organized as follows.
In Section \ref{sec:cusum}, the limiting distribution of the multivariate CUSUM process is derived under both the null hypothesis and local alternatives.
Section \ref{sec:backCUSUM} introduces the backward CUSUM and the stacked backward CUSUM tests for both retrospective testing and monitoring.
While the backward CUSUM is only defined for $t \leq T$ and can thus be implemented only for retrospective testing, the stacked backward CUSUM cumulates recursive residuals backwardly in a triangular scheme and is therefore suitable for real-time monitoring.
The local powers of the tests are compared in Section \ref{sec:localpower}.
In the retrospective setting, the powers of the backward CUSUM and the stacked backward CUSUM tests are substantially higher than that of the conventional forward CUSUM test if a single break occurs after one third of the sample size.
In the case of monitoring, the detection delay of the stacked backward CUSUM under local alternatives is shown to be much lower than that of the monitoring CUSUM detector by \cite{chu1996}.
In Section \ref{sec:infinitehorizon} we present a strong invariance principle for the multivariate CUSUM process and propose an infinite horizon monitoring procedure.
Section \ref{sec:breakpoint} considers the estimation of the break date based on backwardly cumulated recursive residuals.
We present an estimator, which is more accurate than the conventional maximum likelihood estimator if the break is located at the end of the sample.
Section \ref{sec:simulations} presents Monte Carlo simulation results, in Section \ref{sec:covid} we provide a real-data example on monitoring SARS-CoV-2 infections during the COVID-19 pandemic, and Section \ref{sec:Conclusion} concludes.

Throughout the paper, we use the following notation: $\| a \| = \max_{i=1,\ldots, k} |a_i|$ denotes the maximum norm
and $\| A \|_M = \max_{i=1, \ldots, k} \sum_{j=1}^l |A_{i,j}|$ denotes the maximum absolute row sum norm, where  $a \in \mathbb R^k$, and $A \in \mathbb R^{k \times l}$.
We use $\pplim$ to denote convergence in probability as $T \to \infty$, $\Dlim$ for convergence in distribution, and $\Deq$ to indicate that two random variables have the same distribution.
The space of right continuous functions with left limits (c\`{a}dl\`{a}g) on $[0,m]$, where $0 < m < \infty$, is denoted as $D([0,m])$, and its $k$-fold product space is $D([0,m])^k = D([0,m]) \times \ldots \times D([0,m])$.
The space is equipped with the Skorokhod metric (see \citealp{billingsley1999}), and the symbol ``$\Rightarrow$'' denotes weak convergence with respect to this metric.

%% file: 2_CUSUMprocess.tex
\section{The multivariate CUSUM process}	\label{sec:cusum}

We consider the multiple linear regression model
\begin{equation}
	y_t = x_t' \beta_t + u_t, \quad t \in \mathbb{N},	\label{eq:model}
\end{equation}
where $y_t$ is the dependent variable, and $x_t = (1, x_{t2}, \dots, x_{tk})^\T$ is the vector of regressor variables including a constant.
The $k \times 1$ vector of regression coefficients $\beta_t$ depends on the time index $t$, and $u_t$ is an error term.
The time point $T$ divides the time horizon into the retrospective time period $t \leq T$ and the monitoring period $t > T$.
We impose the following assumptions on the regressors and the error term.
\begin{assumption} \
\begin{itemize}
	\item[(a)] The errors satisfy $E[u_t] = 0$, $E[u_t^2]=\sigma^2 > 0$, and $E[|u_t|^8] < \infty$ for all $t$.
	\item[(b)] The regressors satisfy $E[\| x_t \|^8] < \infty$ for all $t$, and the sample covariance matrices $\widehat C_t = t^{-1} \sum_{j=1}^t x_j x_j'$ are uniformly positive definite for all $t > k$ with $\plim_{T \to \infty} \widehat C_T = C$.
	\item[(c)] There exists a positive definite $\Omega$ such that 
	$\plim_{T \to \infty} T^{-1}( \sum_{t=1}^T x_t u_t) ( \sum_{t=1}^T x_t u_t )' = \Omega$.
\end{itemize}
\label{ass:model1} \label{assumption}
\end{assumption}

\noindent
Model \eqref{eq:model} allows for conditionally heteroskedastic errors and local non-stationary regressors, provided that a global long-run covariance matrix $\Omega$ exists.
The regressors can contain lagged dependent variables such as in autoregressive distributed lag models.
We focus on models with a correctly specified dynamic structure and uncorrelated errors.

\begin{assumption}
The error process $u_t$ is a martingale difference sequence with respect to $\mathcal F_t$, the $\sigma$-algebra generated by $\{( x_{i+1}', u_{i})', \ i \leq t\}$.
\label{ass:md}
\end{assumption}

Following \cite{brown1975} this assumption rules out autocorrelated error processes. In practice this may require a dynamic specification with a suitable lag distribution of the variables. 
In Remark \ref{rem:serialcorrelation} we show that autocorrelated errors can be accommodated by replacing the ordinary covariance matrix by a (consistent estimate of the) long-run covariance matrix. Since the estimation of long-run covariances can lead to finite sample size distortions  (see e.g.\ \citealt{casini2021b}), Assumption \ref{ass:md} is a common and convenient assumption in practice.
The expression of the global covariance matrix simplifies to $\Omega = \sigma^2 C$ under Assumption \ref{ass:md}.

Recursive residuals for linear regression models were introduced by \cite{brown1975} as standardized one-step ahead forecast errors, and are defined as
\begin{equation*}
	w_t = \frac{y_t -  x_t^\T  \widehat \beta_{t-1}}{ ( 1 + x_t^\T ( \sum_{i=1}^{t-1}  x_i  x_i^\T )^{-1}  x_t)^{1/2} }, \quad t \geq k + 1,
\end{equation*}
and $w_t = 0$ for $t = 1, \ldots, k$,
where $\widehat{ \beta}_{t-1} = (\sum_{i=1}^{t-1}  x_i  x_i^\T )^{-1} \sum_{i=1}^{t-1} x_i y_i$. Using recursive residuals instead of ordinary OLS residuals as in \cite{ploberger1992} has a number of advantages. 
First, the recursive residuals behave exactly as under the null hypothesis until the parameters change, whereas a structural break affects all OLS residuals in a different manner.
Second, under Assumptions \ref{ass:model1} and \ref{ass:md}, the recursive residuals form a martingale difference sequence regardless of the estimation error in the recursive residuals. 
By contrast, the OLS residuals are (slightly) autocorrelated, which only disappears if the sample size gets large.

The conventional CUSUM detector is given by $S_{t,T} = \widehat \sigma_T^{-1} T^{-1/2} \sum_{i=1}^t w_i$, where $\widehat \sigma_T^2$ denotes the sample variance of $\{ w_{k+1}, \ldots, w_T\}$.
Under the null hypothesis $H_0: \beta_t = \beta_0$ for all $t$, the univariate CUSUM process obeys the functional central limit theorem $S_{\rT,T} \Rightarrow W(r)$, where  $W(r)$ is a standard Brownian motion (see \citealp{sen1982}).
The univariate CUSUM {\it$R$-test} of \cite{brown1975} rejects the null hypothesis if the path of $|S_{t,T}|$ exceeds the linear critical boundary function $b_t = \lambda_\alpha  d_{\text{lin}}(t/T)$ for at least one time index $t=1, \ldots, T$, where
\begin{equation}
	d_{\text{lin}}(r) = 1+2r. \label{linearboundary}
\end{equation}
The critical value $\lambda_\alpha$ is the $(1-\alpha)$ quantile of $\sup_{0\leq r \leq 1} |W(r)|/d_{\text{lin}}(r)$ and determines the significance level $\alpha$, which accounts for the multiplicity issue of the sequential test procedure.
In the monitoring context, \cite{chu1996} considered the radical type boundary function
$b_{\text{rad}}(r) = r^{1/2}  (\log(r) - \log(\alpha^2))^{1/2}$,
which is derived from the boundary crossing probability for a Brownian motion (see \citealp{robbins1970}).
The conventional univariate CUSUM {\it$M$-test} rejects the null hypothesis if the detector statistic $|S_{t,T} - S_{T,T}|$ exceeds $b_t = b_{\text{rad}}(t/T)$ for some $t > T$.

A weakness of univariate CUSUM tests is that they focus on breaks in the intercept. \cite{ploberger1990} studied local alternatives of the form $\beta_t =  \beta_0 + T^{-1/2} g(t/T)$, where $g: \mathbb R \to \mathbb R^k$ is piecewise constant and bounded.
The authors showed that $S_{\rT ,T} \Rightarrow W(r) + \pi^\T h(r)$, where $\pi = e_1' C$, $e_1 = (1,0, \ldots, 0)'$, and
\begin{equation}
	h(r) = \frac{1}{\sigma} \int_0^r g(z) \dd z - \frac{1}{\sigma} \int_0^r \int_0^z \frac{1}{z} g(v) \dd v \dd z.	\label{f_h(r)}
\end{equation}
Consequently, univariate CUSUM tests have no power if $g(r)$ is orthogonal to $\pi$.
To sidestep this difficulty, we follow \cite{jiang2019} and consider the multivariate statistic
\begin{equation*}
	Q_T(r) =\frac{1}{\widehat \sigma_T \sqrt{T}} \widehat C_T^{-1/2} \sum_{t=1}^{\rT} x_t w_t.
\end{equation*}
Under Assumption 1, the multivariate series $x_t u_t$ obeys a multivariate functional central limit theorem (see \citealt{phillips1986}), which also applies to the multivariate CUSUM process of recursive residuals.

\begin{theorem} \label{thm:fclt} \label{thm_fclt_H1}
Let Assumptions \ref{ass:model1} and \ref{ass:md} hold true. If $\beta_t = \beta_0$ for all $t$, then
\begin{equation}
	Q_T(r) \Rightarrow W^{(k)}(r),	\quad r \in [0,m],		\label{eq:fcltH0}
\end{equation}
for any $m < \infty$, as $T \to \infty$, where $W^{(k)}(r)$ is a $k$-dimensional standard Brownian motion.
If $\beta_t = \beta_0 + T^{-1/2} g(t/T)$, where $g(r)$ is piecewise constant and bounded, then
\begin{equation*}
	Q_T(r) \Rightarrow W^{(k)}(r) + C^{1/2} h(r),	\quad r \in [0,m].
\end{equation*}
\end{theorem}

This is an extension of the results in \cite{jiang2019}, who considered slightly stronger assumptions and no local alternatives.
Note that the function $g(r)$ is constant if and only if $\beta_t$ is constant.
If $\beta_t = \beta_0$ for all $t$, we have $h(r) = 0$.
By contrast, under a local alternative with a non-constant break function $g(r)$, it follows that $h(r)$ is non-zero, and, consequently, $C^{1/2} h(r)$ is non-zero, since $C^{1/2}$ is positive definite.
Hence, sequential tests that are based on $Q_T(r)$ have power against a larger class of alternatives than the tests of \cite{brown1975} and \cite{chu1996}.

Therefore, we consider {\it$R$-} and {\it$M$-tests} that are based on the multivariate detector $Q_{t,T} = Q_T(t/T)$.
Note that $Q_{t,T} = S_{t,T}$ if there is only an intercept in the model.
The multivariate forward CUSUM {\it$R$-test} is defined by the following rule: the null hypothesis is rejected if the path of $\| Q_{t,T} \|$ exceeds the boundary function $b_t = \lambda_\alpha d(t/T)$ for at least one index $t=1, \ldots, T$.
Equivalently, we can express this sequential test as a one-shot test, where $H_0$ is rejected if the maximum statistic $\mathcal Q_T = \max_{t=1, \ldots, T} \|Q_{t,T} \|/d(t/T)$ exceeds the critical value $\lambda_\alpha$, which is the $(1-\alpha)$ quantile of its limiting null distribution.

\begin{assumption} \label{ass:boundary}
The boundary function is of the form $b(r) = \lambda_\alpha d(r)$, where $d(r)$ is continuous. There exists $\epsilon > 0$ such that $d(r) > \epsilon$ for all $r \geq 0$.
\end{assumption}

\noindent
By Theorem \ref{thm:fclt} and the continuous mapping theorem it follows that
\begin{align*}
	\mathcal Q_T \Dlim \sup_{r \in (0,1)} \frac{\| W^{(k)}(r) \|}{ d(r)}
\end{align*}
under the null hypothesis.
The multivariate forward CUSUM {\it $M$-test} with fixed endpoint $M=\mT$ rejects $H_0$ if the path of $\| Q_{t,T} - Q_{T,T} \|$ exceeds the boundary function $b_t= \lambda_\alpha d((t-T)/T)$ for at least one index $t=T+1, \ldots, \mT$, where $1 < m < \infty$.
The corresponding maximum statistic is $\mathcal Q_{T,m} = \max_{t=T+1, \ldots, \mT} \|Q_{t,T} - Q_{T,T} \|/d((t-T)/T)$, where, under $H_0$,
\begin{align*}
	\mathcal Q_{T,m} \Dlim \sup_{ r \in (0, m-1)} \frac{\| W^{(k)}(r) \|}{d(r)}.
\end{align*}

\begin{remark} \label{rem:serialcorrelation}
If the dynamics of the model are not specified correctly, the errors may be autocorrelated and Assumption \ref{ass:md} does not apply.
In this case, the limiting distribution differs from that in \eqref{eq:fcltH0} and depends on the global long-run covariance matrix $\Omega$.
Under additional strong mixing assumptions, the process $x_t u_t$ obeys the multivariate functional central limit theorem $T^{-1/2} \sum_{t=1}^{\rT} x_t u_t \Rightarrow \Omega^{1/2} W^{(k)}(r)$ (see \citealt{wooldridge1988}).
To obtain the same limiting distribution as in Theorem \ref{thm:fclt}, we may consider the modified multivariate CUSUM detector
\begin{equation*}
	\widetilde Q_T(r) =\frac{1}{\sqrt{T}} \widehat \Omega_T^{-1/2} \sum_{t=1}^{\rT} x_t w_t, \quad r \in [0,m],
\end{equation*}
where $\widehat \Omega_T$ is some consistent estimator for $\Omega$.
Suitable choices are the long-run covariance estimators of \cite{newey1987} and \cite{andrews1991}.
An alternative is the double kernel HAC estimator by \cite{casini2021b}, which performs well in the presence of locally stationary regressors.
In the Appendix we show that, if Assumption \ref{ass:model1} holds and if there exists $\kappa \geq 8$ such that $\sup_t E[\| x_t \|^\kappa] < \infty$, $\sup_t E[|u_t|^\kappa] < \infty$, and $(x_t, u_t)'$ is strong mixing of size $-\kappa/(\kappa-6)$, then, under the null hypothesis, $\widetilde Q_T(r) \Rightarrow W^{(k)}(r)$, as $T \to \infty$.
Therefore, all {\it $R$-} and {\it $M$-tests} can also be constructed based on the modified detector $\widetilde Q_{t,T} = \widetilde Q_T(t/T)$.
A second approach to deal with possible autocorrelation was proposed by \cite{robbins2011}. Their two-step adjustment approach first employs an ARMA model in order to obtain the pre-whitened residuals which in turn replaces the original residuals in the detector. In our case the ARIMA pre-whitening may be performed in a recursive fashion. As shown by \cite{robbins2011} the resulting detector possesses similar asymptotic properties as the original detector (apart from a scaling factor that depends on the long-run variance).
\end{remark}

\begin{remark} \label{rem:partial}
In practice, partial or one-sided tests can be beneficial in terms of a more powerful test if one is interested in breaks in certain coefficients or directions.
For testing the partial hypothesis $H_0: H' \beta_t = H' \beta_0$, where $H$ is a $k \times l$ matrix with full column rank, we consider the partial CUSUM process $Q_T^*(t/T) = Q_{t,T}^* = \widehat \sigma_T^{-1} T^{-1/2} (H' C_T H)^{-1/2} H' \sum_{j=1}^t x_j w_j$.
All {\it $R$-} and {\it $M$-tests} can be defined with respect to $Q_{t,T}^*$, where $Q_T^*(r) \Rightarrow W^{(l)}(r)$, under $H_0$ and the conditions of Theorem \ref{thm:fclt}.
The {\it $R$-test} by \cite{brown1975} and the {\it $M$-test} by \cite{chu1996} are partial structural break tests for which the matrix $H$ coincides with the first unit vector.
In case of one-sided tests, e.g.\ $H_1: H' \beta_t > H' \beta_0$, the maximum norm can be replaced by the simple maximum, so that $H_0$ is rejected if $p(H' Q_{t,T})$ exceeds the respective boundary function, where $p(x) = \max_{i=1, \ldots, l} x_i$, $x \in \mathbb R^l$.
\end{remark} 

%% file: 3_CUSUMtests.tex
\section{Backward CUSUM R- and M-tests}	\label{sec:backCUSUM}

\subsection{Backward CUSUM R-test}

\begin{table}[t]
\caption{Asymptotic critical values for $\mathcal{Q}_T$ and $\mathcal{BQ}_T$}
\centering
\begin{footnotesize}
\begin{tabular}{l|llllllllll}
 & $k=1$ & $k=2$ & $k=3$ & $k=4$ & $k=5$ & $k=6$ & $k=7$ & $k=8$ & $k=9$ & $k=10$\\ \hline 
$\alpha=10\%$ & 0.848 & 0.944 & 0.996 & 1.031 & 1.058 & 1.080 & 1.097 & 1.112 & 1.125 & 1.138 \\ 
$\alpha=5\%$ & 0.947 & 1.034 & 1.082 & 1.115 & 1.141 & 1.161 & 1.177 & 1.190 & 1.203 & 1.214 \\ 
$\alpha=1\%$ & 1.144 & 1.219 & 1.258 & 1.283 & 1.303 & 1.324 & 1.343 & 1.357 & 1.368 & 1.381 \\ \hline
 \hline
\end{tabular}
\end{footnotesize}
\hspace*{0.1ex}
\parbox{14.5cm}{ \vspace{0.5ex} \spacingset{1} \scriptsize
Note: Simulated $(1-\alpha)$ quantiles of $\sup_{r \in (0,1)} \| W^{(k)}(r)\|/(1+2r)$ based on 100,000 Monte Carlo replications are reported. 
The Wiener process is approximated on a grid of 50,000 equidistant points.
}
\label{tab_crit1}
\end{table}

\begin{figure}[t]
\centering
\caption{Illustrative example for the backward CUSUM with a break in the mean}
\includegraphics[scale=0.37]{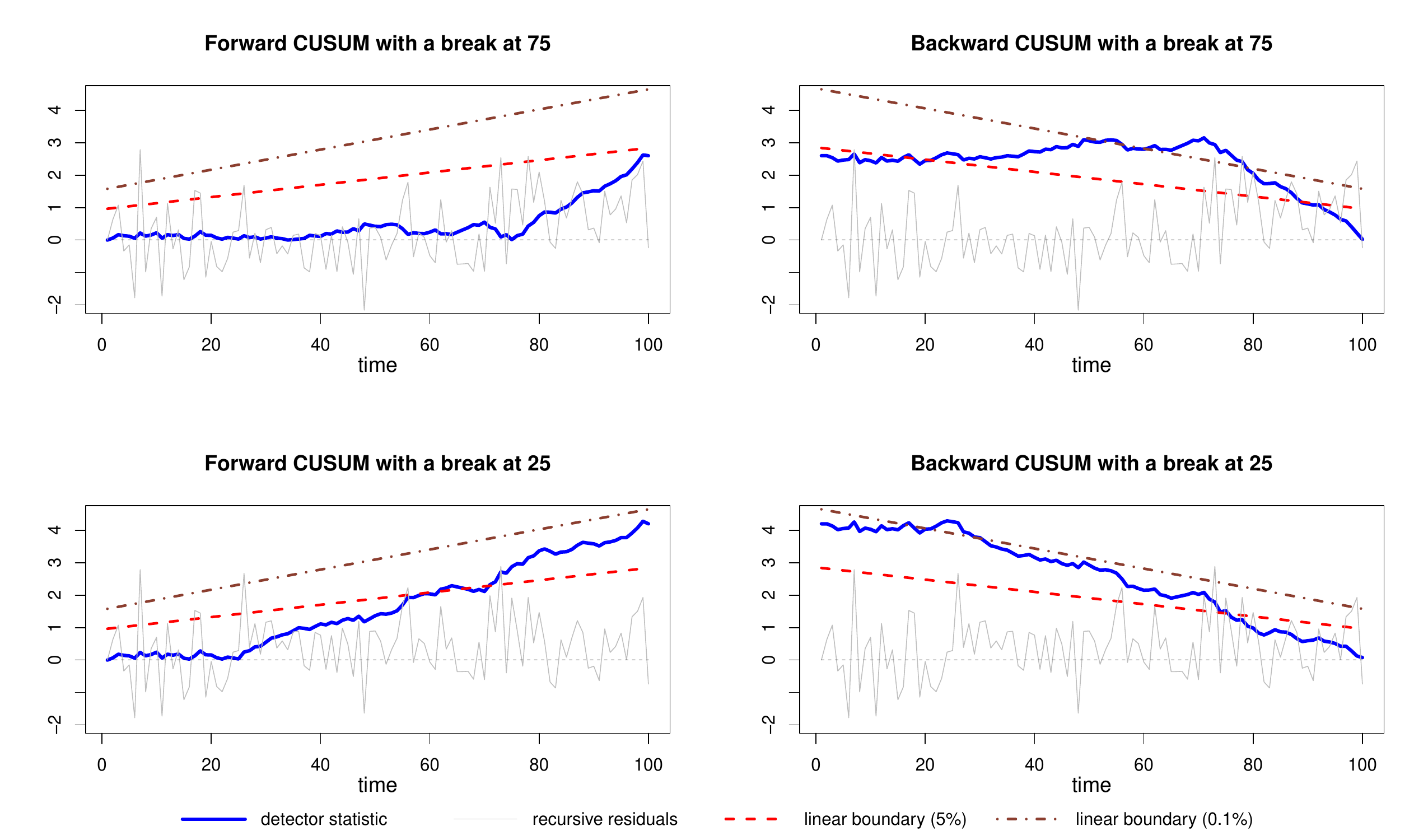}
\parbox{15cm}{ \vspace*{0.5ex} \spacingset{1} \scriptsize
Note: The process $y_t = \mu_t + u_t$, $t=1, \ldots, T$, is simulated for $T=100$ with $\mu_t = 0$ for $t < \tau^* T$, $\mu_t = 1$ for $t \geq \tau^* T$, and $u_t \sim NID(0,1)$, where $\tau^* = 0.75$ in the upper panels and $\tau^* = 0.25$ in the lower panels.
The bold solid line paths are the trajectories of $\|Q_{t,T}\|$ and $\|BQ_{t,T}\|$, where the detectors are univariate such that the norm is just the absolute value.
In the background, the recursive residuals are plotted.
The dashed and dash-dotted lines correspond to the linear boundary
$d_\text{lin}(r)$ with significance levels $\alpha = 5\%$ and $\alpha = 0.1\%$, respectively.
}
\label{fig_motivation}
\end{figure}

An alternative approach is to cumulate the recursive residuals in reversed order.
Suppose there is a single break in $\beta_t$ at time $t=T^*$.
Then, $\{ w_t, \ t < T^* \}$ are the residuals from the pre-break period, and $\{ w_t, \ t \geq T^* \}$ are those from the post-break period.
As the pre-break recursive residuals are not affected by a violation of the null hypothesis, they do not provide useful information about a subsequent break. Accordingly, the partial sum process $T^{-1/2} \sum_{j=1}^t w_j$ behaves like a pure random walk for $t < T^*$ and cumulating those residuals brings nothing but noise to the detector statistic. In contrast, the post-break residuals have nonzero mean and reveal relevant information about a possible break. In order to focus on the post-break residuals, we therefore consider backwardly cumulated partial sums of the form $T^{-1/2} \sum_{j=0}^{t-1} w_{T-j}$. We define the retrospective backward CUSUM detector as
\begin{equation*}
	BQ_{t,T} = Q_T(1) - Q_T\big(\tfrac{t-1}{T}\big) = \frac{1}{\widehat \sigma_T \sqrt{T}} C_T^{-1/2} \sum_{j=t}^T x_j w_j	\quad (t=1, \ldots, T).
\end{equation*}
The null hypothesis is rejected if $\|BQ_{t,T}\|$ exceeds the boundary
$b_t = \lambda_\alpha d((T-t-1)/T)$ for at least one time index $t$. 
The maximum statistic is given by
\begin{align*}
	\mathcal{BQ}_T = \max_{t= 1, \ldots, T} \frac{\|BQ_{t,T}\|}{d(\tfrac{T-t+1}{T})},
\end{align*}
and, under the local alternatives defined in Theorem \ref{thm:fclt}, the continuous mapping theorem implies
\begin{align}
\mathcal{BQ}_T \Dlim &\sup_{r \in (0,1)} \frac{\|W^{(k)}(1) + C^{1/2} h(1) - W^{(k)}(r) - C^{1/2} h(r)\|}{d(1-r)} \nonumber \\
&\Deq \sup_{r \in (0,1)} \frac{\| W^{(k)}(r) + C^{1/2} (h(1) - h(1-r)) \|}{d(r)}. \label{eq:BQdistribution}
\end{align}
Hence, the limiting distribution of $\mathcal{BQ}_T$ under $H_0$ coincides with that of $\mathcal{Q}_T$. 
Simulated asymptotic critical values under the linear boundary \eqref{linearboundary} are presented in Table \ref{tab_crit1}.
Under local alternatives, the limiting distributions of $\mathcal{BQ}_T$ and $\mathcal{Q}_T$ differ. 
A simple illustrative example of the detector paths together with the linear boundary \eqref{linearboundary} of \cite{brown1975} are depicted in Figure \ref{fig_motivation}, in which two processes with $k=1$ and a single break in the mean at $3/4$ and $1/4$ of the sample are simulated.

Unlike the forward CUSUM detector, the backward CUSUM detector is not measurable with respect to the filtration of available information at time $t$ and is therefore not suitable for a monitoring procedure.
The path of $\|BQ_{t,T}\|$ cannot be monitored in real-time, as it is only defined for $t \leq T$ with fixed endpoint $T$.
To obtain a feasible {\it $M$-test} in practice, we resort to a triangular backward inspection scheme of recursive residuals, which is discussed below.

\subsection{Stacked backward CUSUM R-test}

Let $\mathcal{BQ}_T(t) = \max_{s=1, \ldots, t} \|Q_T(t/T) -  Q_T((s-1)/T)\|/d((t-s+1)/T)$ be the backward CUSUM maximum statistic with endpoint $t$.
The idea of the stacked backward CUSUM scheme is to compute this statistic sequentially for each $t$, yielding
$\mathcal{BQ}_T(1), \ldots, \mathcal{BQ}_T(T)$.
The corresponding maximum statistic $\mathcal{SBQ}_T$ is the maximum among this sequence of backward CUSUM statistics.
An important feature is that this sequence is measurable with respect to the filtration of information at time $t$, so that $\mathcal{BQ}_T(t)$ is itself a sequential statistic.
Stacking all backward CUSUM statistics on one another leads to a triangular array structure given by
	\begin{equation}
		SBQ_{s,t,T} =  Q_T\big(\tfrac{t}{T}\big) -  Q_T\big(\tfrac{s-1}{T}\big) = \frac{1}{\widehat \sigma_T \sqrt{T}} C_T^{-1/2}\sum_{j=s}^t x_j w_j \quad  (t \in \mathbb N, \ s=1, \ldots, t),	\label{eq:SBQdetector}
	\end{equation}
which is denoted as the stacked backward CUSUM detector.
We reject $H_0$ if $\|SBQ_{s,t,T}\|$ exceeds the triangular boundary $b_{s,t} = b(\frac{t}{T}, \frac{s-1}{T})$ for some $t=1, \ldots, T$ and $s = 1, \ldots, t$.
\begin{assumption} \label{ass:boundarystacked}
The triangular boundary function is of the form $b(r,s) = \lambda_\alpha d(r,s)$, where $d(r,s)$ is continuous. There exists $\epsilon > 0$ such that $d(r,s) > \epsilon$ for all $0 \leq s \leq r$.
\end{assumption}

\noindent
The stacked backward CUSUM {\it $R$-test} can be equivalently expressed in terms of a maximum statistic. $H_0$ is rejected if the double maximum statistic
\begin{equation*}
	\mathcal{SBQ}_T = \max_{t=1, \ldots, T} \mathcal{BQ}_T(t) = \max_{t=1, \ldots, T} \max_{s=1, \ldots, t} \frac{\| SBQ_{s,t,T} \|}{d( \tfrac{t}{T}, \tfrac{s-1}{T})}
\end{equation*}
exceeds $\lambda_\alpha$.
Under the local alternatives defined in Theorem \ref{thm:fclt}, it follows that
\begin{align*}
	\mathcal{SBQ}_T  \Dlim  \sup_{r \in (0,1)} \sup_{s \in (0,r)} \frac{\|W^{(k)}(r) - W^{(k)}(s) + C^{1/2} [h(r) - h(s)] \|}{d(r,s)}.
\end{align*}

\begin{table}[t]
\caption{Asymptotic critical values for $\mathcal{SBQ}_{T,m}$ (fixed $m$)}
\centering
\begin{footnotesize}
\begin{tabular}{l|cccccccccccc}
 & \multicolumn{3}{c}{$k=1$} & \multicolumn{3}{c}{$k=2$} & \multicolumn{3}{c}{$k=3$} & \multicolumn{3}{c}{$k=4$} \\ \hline
\multicolumn{1}{c|}{$\alpha$} & $10\%$ & $5\%$ & $1\%$ & $10\%$ & $5\%$ & $1\%$ & $10\%$ & $5\%$ & $1\%$ & $10\%$ & $5\%$ & $1\%$ \\
$m=1.2$ & 0.780 & 0.859 & 1.023 & 0.857 & 0.932 & 1.082 & 0.900 & 0.973 & 1.121 & 0.930 & 1.002 & 1.147 \\ 
$m=1.4$ & 0.944 & 1.030 & 1.208 & 1.026 & 1.107 & 1.270 & 1.073 & 1.153 & 1.316 & 1.107 & 1.183 & 1.345 \\ 
$m=1.6$ & 1.024 & 1.114 & 1.290 & 1.109 & 1.189 & 1.356 & 1.156 & 1.235 & 1.398 & 1.190 & 1.266 & 1.428 \\ 
$m=1.8$ & 1.077 & 1.166 & 1.341 & 1.161 & 1.241 & 1.406 & 1.207 & 1.285 & 1.446 & 1.241 & 1.318 & 1.476 \\ 
$m=2$ & 1.116 & 1.202 & 1.374 & 1.195 & 1.274 & 1.438 & 1.243 & 1.319 & 1.479 & 1.275 & 1.351 & 1.506 \\ 
$m=4$ & 1.268 & 1.346 & 1.510 & 1.342 & 1.414 & 1.567 & 1.386 & 1.455 & 1.600 & 1.415 & 1.483 & 1.625 \\ 
$m=10$ & 1.392 & 1.462 & 1.610 & 1.460 & 1.527 & 1.665 & 1.499 & 1.564 & 1.695 & 1.526 & 1.589 & 1.722 \\ 
 & \multicolumn{3}{c}{$k=5$} & \multicolumn{3}{c}{$k=6$} & \multicolumn{3}{c}{$k=7$} & \multicolumn{3}{c}{$k=8$} \\ \hline
\multicolumn{1}{c|}{$\alpha$} & $10\%$ & $5\%$ & $1\%$ & $10\%$ & $5\%$ & $1\%$ & $10\%$ & $5\%$ & $1\%$ & $10\%$ & $5\%$ & $1\%$ \\
$m=1.2$ & 0.953 & 1.021 & 1.167 & 0.971 & 1.038 & 1.182 & 0.986 & 1.052 & 1.194 & 0.999 & 1.065 & 1.205 \\ 
$m=1.4$ & 1.131 & 1.206 & 1.363 & 1.151 & 1.225 & 1.378 & 1.167 & 1.240 & 1.390 & 1.180 & 1.253 & 1.402 \\ 
$m=1.6$ & 1.214 & 1.290 & 1.446 & 1.235 & 1.310 & 1.461 & 1.251 & 1.324 & 1.473 & 1.264 & 1.337 & 1.486 \\ 
$m=1.8$ & 1.265 & 1.340 & 1.493 & 1.285 & 1.360 & 1.512 & 1.301 & 1.374 & 1.525 & 1.314 & 1.387 & 1.538 \\ 
$m=2$ & 1.299 & 1.374 & 1.529 & 1.318 & 1.392 & 1.544 & 1.334 & 1.407 & 1.555 & 1.347 & 1.419 & 1.565 \\ 
$m=4$ & 1.436 & 1.504 & 1.644 & 1.453 & 1.522 & 1.659 & 1.469 & 1.536 & 1.673 & 1.482 & 1.548 & 1.683 \\ 
$m=10$ & 1.546 & 1.608 & 1.739 & 1.563 & 1.624 & 1.755 & 1.576 & 1.638 & 1.765 & 1.587 & 1.649 & 1.774 \\ \hline \hline
\end{tabular}
\end{footnotesize}
\hspace*{0.1ex}
\parbox{16.2cm}{ \vspace{0.5ex} \spacingset{1} \scriptsize
Note: Simulated $(1-\alpha)$ quantiles of $\sup_{r \in (0,m-1)} \sup_{s \in (0,r)} \| W^{(k)}(r) - W^{(k)}(s)\|/(1+2(r-s))$ based on 100,000 Monte Carlo replications are reported. 
The Wiener process is approximated on a grid of 50,000 equidistant points.
The critical values for $\mathcal{SBQ}_{T}$ ($R$-test) coincide with those for $\mathcal{SBQ}_{T,m}$ ($M$-test) with $m=2$.
}
\label{tab_crit2}
\end{table}

\subsection{Stacked backward CUSUM M-test}

Since the triangular detector \eqref{eq:SBQdetector} is measurable with respect to the information set at time $t$, it can be monitored on-line across all time points $t > T$.
The null hypothesis is rejected if $\|SBQ_{s,t,T}\|$ exceeds $b_{s,t} = b(\frac{t}{T}, \frac{s-1}{T})$ at least once for some $t \geq T+1$ and $s = T+1, \ldots, t$.
The {\it $M$-test} maximum statistic with a fixed horizon $m < \infty$ is given by 
\begin{equation*}
	\mathcal{SBQ}_{T,m} = \max_{t=T+1, \ldots, \mT} \max_{s=T+1, \ldots, t} \frac{\| SBQ_{s,t,T} \|}{d( \frac{t}{T}, \frac{s-1}{T} )},
\end{equation*}
where, analogously to \eqref{eq:BQdistribution},
\begin{align*}
	\mathcal{SBQ}_{T,m} \Dlim& \sup_{r \in (1,m)} \sup_{s \in (1,r)} \frac{\|W^{(k)}(r) - W^{(k)}(s) + C^{1/2}[h(r) - h(s)] \|}{d(r,s)} \\
	\Deq& \sup_{r \in (0,m-1)} \sup_{s \in (0,r)} \frac{\|W^{(k)}(r) - W^{(k)}(s) + C^{1/2}[h(r+1) - h(s+1)] \|}{d(r,s)}.
\end{align*}
Simulated critical values for the stacked backward CUSUM $R$- and $M$-tests under the linear triangular boundary
\begin{align}
	d_{\text{sbq}}(r,s) = 1+2(r-s)		\label{eq:linearboundarySBQ}
\end{align}
are presented in Table \ref{tab_crit2}.

%% file: 4_localpower.tex
\section{Local power} \label{sec:localpower}

\begin{figure}[t]
\caption{Asymptotic local power curves}
\centering
\includegraphics[scale=0.21]{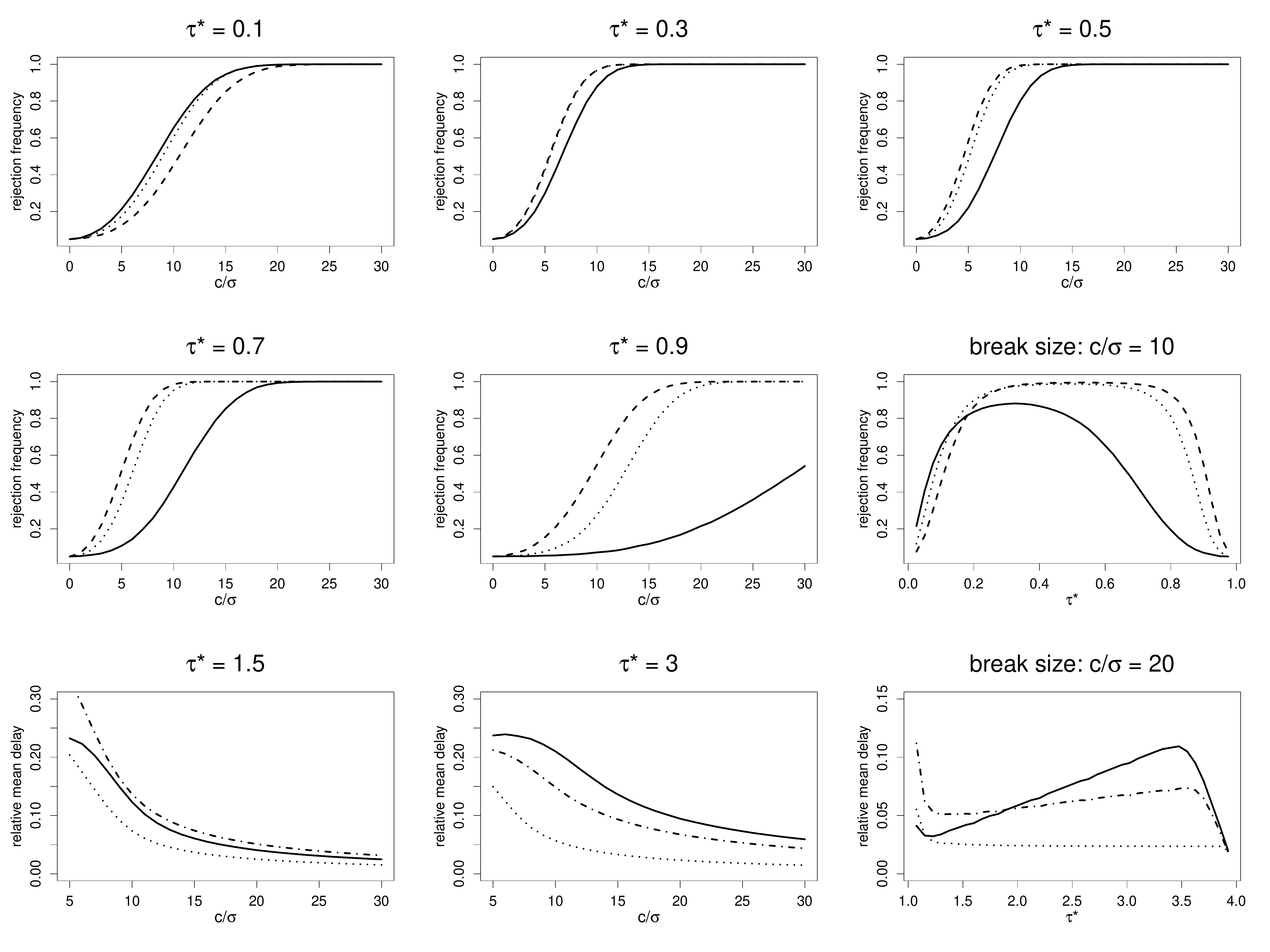}
\parbox{16cm}{ \vspace{0.5ex} \spacingset{1} \scriptsize
Note: The upper six panels show simulated asymptotic local power curves for the {\it $R$-tests} $\mathcal Q_T$ (solid), $\mathcal{BQ}_T$ (dashed) and $\mathcal{SBQ}_T$ (dotted) from equation \eqref{eq:Rtestlimiting}.
The bottom three panels show simulated asymptotic mean local delays for the {\it $M$-tests} $\mathcal{Q}_{T,m}$ (solid) and $\mathcal{SBQ}_{T,m}$ (dotted) from equation \eqref{eq:MtestinglimitQ}, and the test by \cite{chu1996} (dash-dotted), where $m=4$.
The Brownian motions in the limiting distributions are approximated on a grid of 1,000 equidistant points and the rates are obtained from 100,000 Monte Carlo repetitions using size-adjusted $5\%$ critical values.
}
\label{fig_powerplot_singlebreak}
\end{figure}

In order to illustrate the advantages of the backward CUSUM tests, we consider the simple local break model $\beta_t = \beta_0 + T^{-1/2} g(t/T)$ with $g(r) = c 1_{\{ r \geq \tau^* \}}$, where $c \in \mathbb{R}^k$, and $\tau^*$ denotes the break location.
From \eqref{f_h(r)} it follows that
\begin{align}
	h(r) 
	= c \sigma^{-1} \Big( \int_{\tau^*}^r \dd z - \int_0^r \int_{\tau^*}^z \frac{1}{z} \dd v \dd z  \Big) 
	= c \sigma^{-1} \tau^* (\ln(r) - \ln(\tau^*)) 1_{\{ r \geq \tau^* \}},		\label{eq:h(r)}
\end{align}
and, under the linear boundaries \eqref{linearboundary} and \eqref{eq:linearboundarySBQ}, the $R$-tests satisfy
\begin{align}
	&\mathcal Q_T \Dlim \sup_{r \in (0,1)} \frac{\|W(r) + h(r)\|}{1+2r}, \quad 
	\mathcal{BQ}_T \Dlim \sup_{r \in (0,1)} \frac{\|W(r) + h(1)-h(1-r)\|}{1+2r}, \nonumber \\
	&\mathcal{SBQ}_T \Dlim \sup_{r \in (0,1)} \sup_{s \in (0,r)} \frac{\|W(r)-W(s)+h(r)-h(s)\|}{1+2(r-s)}. \label{eq:Rtestlimiting}
\end{align}
Asymptotic local power curves from the limiting distributions in \eqref{eq:Rtestlimiting} for the case $k=1$ are presented in Figure \ref{fig_powerplot_singlebreak}.
The $(2,3)$-element of the panel of figures shows that for a fixed break size the backward CUSUM and the stacked backward CUSUM outperform the forward CUSUM if a single break $\tau^*$ is located after $15\%$ of the sample size.
If the break date $\tau^*$ tends to the end of the sample, the power gain of $\mathcal{BQ}_T$ and $\mathcal{SBQ}_T$ increases substantially.

For the $M$-test statistics with fixed endpoint $m=2$, the limiting distributions  of $\mathcal Q_{T,m}$ and $\mathcal{SBQ}_{T,m}$ for a break at $\tau^* \in (1,2)$ coincide with those for the $R$-tests presented in \eqref{eq:Rtestlimiting} for a break at $\tau^* \in (0,1)$. 
Hence, the power of $\mathcal{SBQ}_{T,m}$ is higher than that of $\mathcal Q_{T,m}$ if breaks are located after 15\% of the pre-monitoring sample.

Another important performance measure for {\it $M$-tests} is the delay between the actual break and the detection time point.
\cite{aue2004} and \cite{aue2009} derived the asymptotic distribution of the detection stopping time of CUSUM $M$-tests that are based on OLS residuals.
For $\mathcal Q_{T,m}$ and $\mathcal{SBQ}_{T,m}$, the detection stopping times are given by
\begin{align*}
	\mathcal D_{T,m,Q} &= \min \bigg\{ t \in \{T+1, \ldots, \mT \} \bigg| \frac{\|Q_{t,T} - Q_{T,T}\|}{  d(\frac{t-T}{T}) \lambda_{\alpha,m,Q}} \geq 1 \bigg\}, \\
	\mathcal D_{T,m,SBQ} &= \min \bigg\{ t \in \{T+1, \ldots, \mT \} \bigg| \max_{T < s \leq t} \frac{\|SBQ_{s,t,T}\|}{d(\frac{t}{T},\frac{s-1}{T}) \lambda_{\alpha,m,SBQ} } \geq 1 \bigg\},
\end{align*}
where $\lambda_{\alpha,m,Q}$ and $\lambda_{\alpha,m,SBQ}$ are the corresponding critical values. 
Under the same setting as in \eqref{eq:Rtestlimiting}, the relative detection stopping times satisfy
\begin{align}
	\frac{\mathcal D_{T,m,Q}}{T} &\Dlim \min  \bigg\{ r \in (0,m-1) \bigg| \frac{\|W(r) + h(r+1) - h(r)\|}{(1+2r) \lambda_{\alpha,m,Q} } \geq 1 \bigg\}, \label{eq:MtestinglimitQ} \\
	\frac{\mathcal D_{T,m,SBQ}}{T} &\Dlim \min \bigg\{ r \in (0,m-1) \bigg| \sup_{s \in (0,r)} \frac{\|W(r)-W(s)+h(r+1)-h(s+1)\|}{ (1+2(r-s)) \lambda_{\alpha,m,Q} } \geq 1 \bigg\}, \nonumber
\end{align}
as $T \to \infty$, where the limiting relative stopping times are denoted as $\tau_{D,Q}$ and $\tau_{D,SBQ}$, respectively.
The asymptotic mean local delays $E ( \tau_{D}  \mid  \tau^* \leq \tau_{D} \leq m ) - \tau^*$, where $\tau_{D} \in \{\tau_{D,Q}, \tau_{D,SBQ} \}$, are presented in the bottom panels of Figure \ref{fig_powerplot_singlebreak} for $m=4$ and different break locations.
The asymptotic mean local delay of $\mathcal{SBQ}_{T,m}$ is much lower than that of $\mathcal{Q}_{T,m}$.
Moreover, the asymptotic mean local delay of $\mathcal{SBQ}_{T,m}$ slowly decreases in $\tau^*$ and is much lower than that of $\mathcal{Q}_{T,m}$, except for early breaks.

\begin{figure}[t]
\caption{Size distributions of the retrospective and monitoring detectors}
\centering
    \mbox{\includegraphics[scale=0.27, page=1]{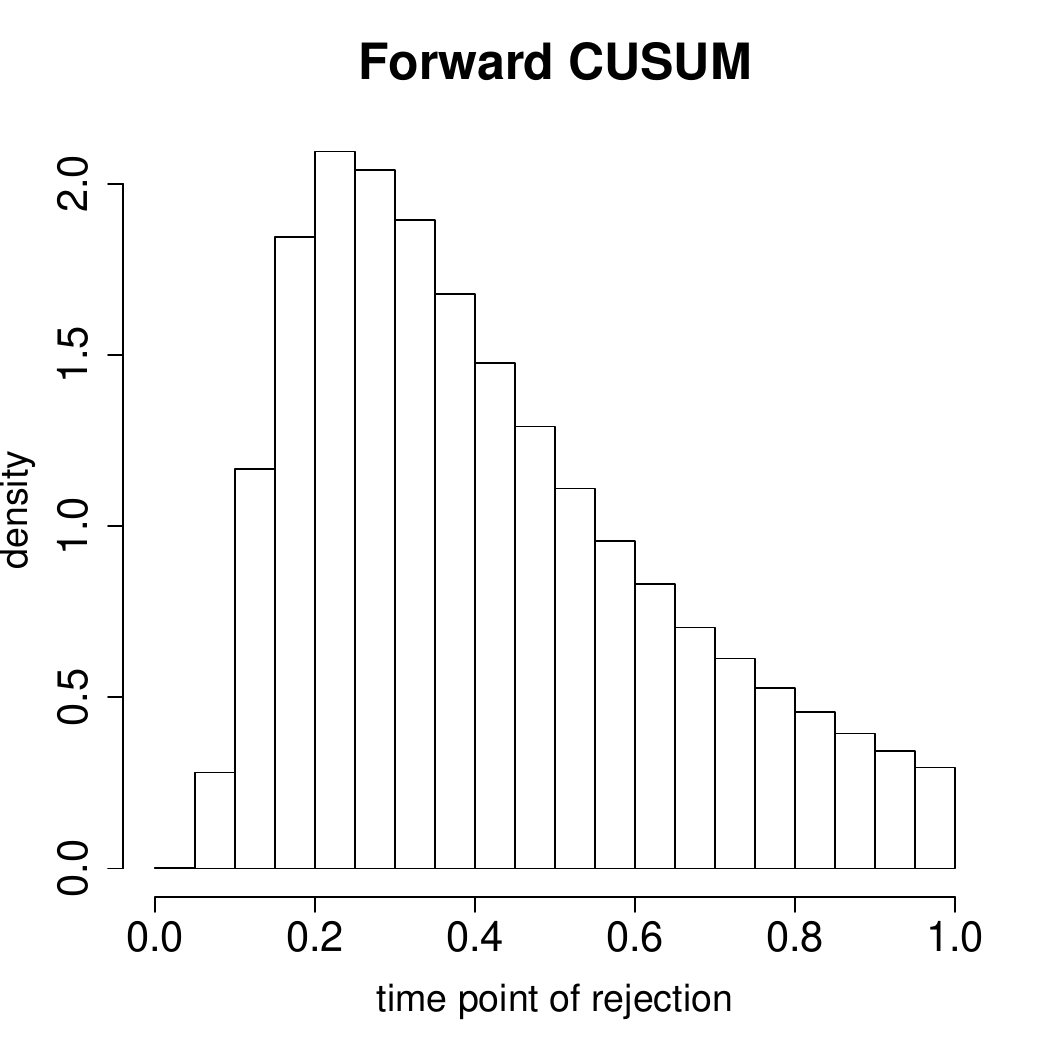}}
    \mbox{\includegraphics[scale=0.27, page=2]{./img/figure4a.pdf}}
    \mbox{\includegraphics[scale=0.27, page=3]{./img/figure4a.pdf}}
        \mbox{\includegraphics[scale=0.27, page=1]{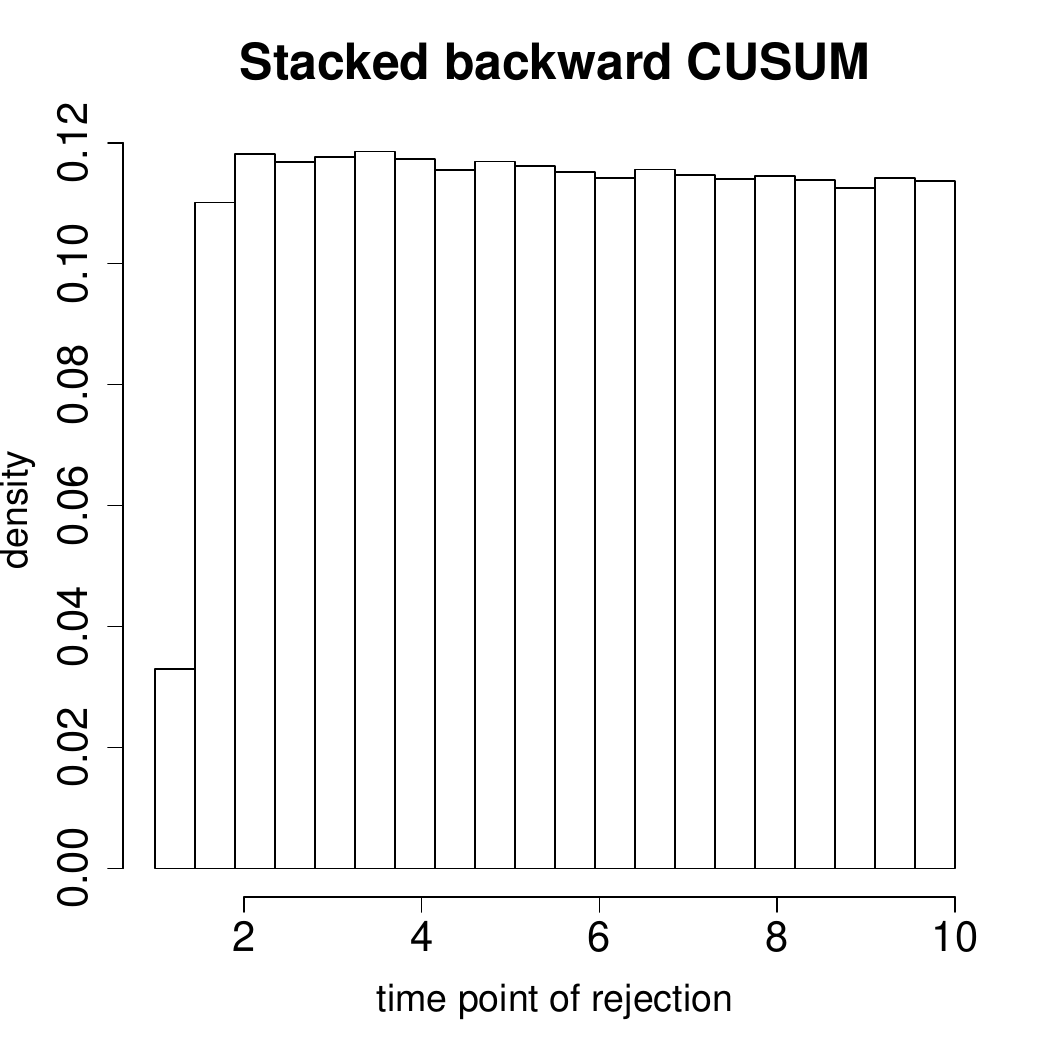}}
    \mbox{\includegraphics[scale=0.27, page=2]{./img/figure4b.pdf}}
    \mbox{\includegraphics[scale=0.27, page=3]{./img/figure4b.pdf}}
\parbox{14.6cm}{ \vspace{0.5ex} \spacingset{1} \scriptsize
Note: The frequencies of the location of the first boundary exceedance under the null hypothesis are shown for a significance level of $5 \%$ for the model with $k=1$.
The frequencies are based on random draws under the limiting null distribution of the maximum statistics.
The retrospective cases are considered for the upper three histograms and the fixed endpoint monitoring case with $m=10$ for the lower three.
The linear boundaries \eqref{linearboundary} and \eqref{eq:linearboundarySBQ} are considered in the first five plots and the radical boundary by \cite{chu1996} is used in the last plot.
}
\label{fig_SizeDist}
\end{figure}

\begin{remark} \label{rem:sizedist}
While, for one-shot tests, the critical value determines the type I error, sequential testing involves two degrees of freedom.
Besides the test size, which is controlled asymptotically by an appropriately chosen value for $\lambda_\alpha$, the shape of the boundary determines the distribution of potential relative crossing time points $r$.
As already noted by \cite{brown1975}, the forward CUSUM with the linear boundary \eqref{linearboundary} puts more weight on detecting breaks that occur early in the sample (c.f. Figure \ref{fig_powerplot_singlebreak}). 
In Figure \ref{fig_SizeDist} we present the distributions of the first boundary crossing under the null hypothesis, which is also referred to as the ``distribution of the size'' (see \citealp{anatolyev2018}).
The results indicate that the size is skewed for the forward and backward CUSUM tests and almost evenly distributed for stacked backward CUSUM tests, which is due to the weighting scheme of the linear triangular boundary function \eqref{eq:linearboundarySBQ}.
There is no consensus on which distribution should be preferred, as whether one wishes to put more weight on particular regions of time points of rejection depends on the particular application.
However, \cite{zeileis2005} and \cite{anatolyev2018} argue that if no further information is available, one might prefer a uniform distribution to a skewed one.
\end{remark}

%% file: 5_infinitehorizon.tex
\section{Infinite horizon monitoring} \label{sec:infinitehorizon}

\begin{table}[t]
\caption{Asymptotic critical values for infinite horizon $M$-tests}
\centering
\begin{footnotesize}
\begin{tabular}{l|lllll|lllll}
\multicolumn{1}{c}{} & \multicolumn{5}{c}{stacked backward CUSUM} & \multicolumn{5}{c}{forward CUSUM} \\
 & $k=1$ & $k=2$ & $k=3$ & $k=4$ & $k=5$ & $k=1$ & $k=2$ & $k=3$ & $k=4$ & $k=5$\\ \hline 
$\alpha=10\%$ & 0.911 & 0.974 & 1.010 & 1.035 & 1.054 & 0.864 & 0.956 & 1.006 & 1.040 & 1.066 \\ 
$\alpha=5\%$ & 0.976 & 1.036 & 1.071 & 1.094 & 1.113 & 0.958 & 1.044 & 1.090 & 1.121 & 1.146 \\ 
$\alpha=1\%$ & 1.113 & 1.169 & 1.199 & 1.219 & 1.236 & 1.148 & 1.222 & 1.261 & 1.289 & 1.308 \\ 
 \hline
 \hline
\end{tabular}
\end{footnotesize}
\hspace*{0.1ex}
\parbox{14.5cm}{ \vspace{0.5ex} \spacingset{1} \scriptsize
Note: Simulated $(1-\alpha)$ quantiles of the limiting distributions presented in Theorem \ref{thm:infinite} based on 100,000 Monte Carlo replications are reported. 
For $\mathcal{SBQ}_{T,\infty}$ we use boundary $d_\text{inf}(r,s)$, and for $\mathcal{Q}_{T,\infty}$ the boundary $d_\text{lin}(r)$ is implemented.
The Brownian bridge is approximated on a grid of 50,000 equidistant points.
}
\label{tab:crit3}
\end{table}

The functional central limit theorem given by Theorem \ref{thm_fclt_H1} is not suitable for analyzing the asymptotic behavior of an infinite horizon monitoring statistic, since the variance of $Q_T(r)$ is unbounded as $r \to \infty$, and $\sup_{r \geq 1} \| Q_T(r) - W^{(k)}(r) \|$ might not converge in general.
Instead, we need an almost sure invariance principle, which is specified by the following high level condition:

\begin{assumption}	\label{ass:strongprinciple}
There exists a $k$-dimensional Brownian motion $W^{(k)}(t)$ such that
\begin{align*}
	\sum_{j=1}^t x_j u_j = \Omega^{1/2} W^{(k)}(t) + o(t^{1/2}), \qquad \text{(a.s.)}.
\end{align*}
\end{assumption}

Almost sure invariance principles were first studied by \cite{strassen1967}, who verified Assumption \ref{ass:strongprinciple} under the additional assumption that $x_t u_t$ is a stationary and ergodic martingale difference sequence.
Optimal rates were first derived by \cite{komlos1975}.
 \cite{aue2004} and \cite{aue2009} present examples where Assumption \ref{ass:strongprinciple} is satisfied, which include martingale difference sequences and linear processes with GARCH-type innovations under mild regularity conditions.
For more general dependent processes under suitable regularity conditions, Assumption \ref{ass:strongprinciple} was shown by \cite{wu2007} and \cite{berkes2014} with respect to a physical dependence measure (see also \citealt{berkes2011}).

\begin{theorem} \label{thm:strongprinciple}
Let Assumptions \ref{ass:model1}, \ref{ass:md}, and \ref{ass:strongprinciple} hold true and let $\beta_t = \beta_0 $ for all $t \in \mathbb{N}$. There exists a $k$-dimensional standard Brownian motion $W^{(k)}(r)$, such that, as $T \to \infty$,
\begin{align*}
	\sup_{r > 1} r^{-1/2} \| Q_T(r) - W^{(k)}(r) \| = o_P(1).
\end{align*}
\end{theorem} 

This result is the key tool to establish the limiting distributions of infinite horizon monitoring statistics under $H_0$ and indicates the need of further restrictions on the boundary function.
In the Appendix, we also show that this result remains valid if we replace $Q_T(r)$ by the autocorrelation robust statistic $\widetilde Q_T(r)$ under additional mixing assumptions.
The infinite horizon forward CUSUM and stacked backward CUSUM maximum statistics are defined as
\begin{align*}
	\mathcal{Q}_{T,\infty} = \max_{t \geq T+1} \frac{\| Q_T(\frac{t}{T}) - Q_T(1) \|}{d( \tfrac{t-T}{T})}, \quad 
	\mathcal{SBQ}_{T,\infty} = \max_{t \geq T+1} \max_{s=T+1, \ldots, t} \frac{\| Q_T(\frac{t}{T}) - Q_T(\frac{s-1}{T}) \|}{d( \frac{t}{T}, \frac{s-1}{T} )}.
\end{align*}
Analogously to the fixed horizon case, $H_0$ is rejected if the statistic exceeds the $(1-\alpha)$ quantile of its limiting null distribution.
Since a maximum over a non-compact set can be unbounded, we need further restrictions on the boundary functions for infinite horizon monitoring.

\begin{assumption} \label{ass:boundary-infinite}
The boundary functions that are defined in Assumptions \ref{ass:boundary} and \ref{ass:boundarystacked} satisfy $\sup_{r > 1} \sqrt r / d(r-1) < \infty$ and $\sup_{r > 1} \sup_{s \in (1,r)} \sqrt r / d(r,s) < \infty$.
\end{assumption}

\noindent
Under this assumption we show the following theorem:
\begin{theorem} \label{thm:infinite}
Let $\beta_t = \beta_0 $ for all $t$, and let Assumptions \ref{assumption}--\ref{ass:boundary-infinite} hold true.
Then, as $T \to \infty$,
\begin{align*}
	\mathcal{Q}_{T,\infty} &\Dlim \sup_{r \in (0,1)} \frac{\| B^{(k)}(r) \|}{(1-r) d(\tfrac{r}{1-r})}, \\
	\mathcal{SBQ}_{T,\infty} 
	&\Dlim \sup_{ r \in (0,1)} \sup_{s \in (0,r)} \frac{\| (1-s) B^{(k)}(r) - (1-r) B^{(k)}(s) \|}{(1-r) (1-s) d(\frac{1}{1-r}, \frac{1}{1-s})},
\end{align*}
where $B^{(k)}(r)$ is a $k$-dimensional standard Brownian bridge.
\end{theorem}

\noindent
The linear boundary $d_\text{lin}(r)$ (see equation \eqref{linearboundary}) satisfies Assumption \ref{ass:boundary-infinite}, whereas for the linear triangular boundary $d_{\text{sbq}}(r,s)$ (see equation \eqref{eq:linearboundarySBQ}) it is not satisfied. 
Instead, the boundary must be at least of order $\sqrt r$ uniformly among all $s$, which motivates the alternative boundary
\begin{align*}
	d_{\text{inf}}(r,s) = \sqrt r(1 + 2(r-s)), \quad 0 \leq s \leq r.
\end{align*}
Simulated critical values for the {\it $M$-tests} under the boundaries $d_\text{lin}(r)$ and $d_\text{inf}(r,s)$ are presented in Table \ref{tab:crit3}.

%% file: 6_breakpoint.tex
\section{Estimation of the breakpoint location}  \label{sec:breakpoint}

As soon as the testing procedure has indicated a structural instability in the coefficient vector, the next step is to locate the break point.
In the single break model with $\beta_t = \beta_0 + \delta 1_{\{t \geq T^*\}}$, where $\delta \neq 0$,
\cite{horvath1995} suggested to estimate the relative break date $\tau^* = T^* / T$ by the relative time index for which the likelihood ratio statistic is maximized.
As an asymptotically equivalent estimator, \cite{bai1997} proposed the maximum likelihood estimator
\begin{equation}
	\widehat \tau^{\text{ret}}_{ML} = T^{-1} \cdot \argmin_{t= 1, \ldots, T} \{ R_1(t) + R_2(t) \},	\label{eq:MLest}
\end{equation}
where $R_1(t)$ is the OLS residual sum of squares when using observations until time point $t$ and $R_2(t)$ is the OLS residual sum of squares when using observations from time $t+1$ onwards.
In case of monitoring, \cite{chu1996} considered
\begin{equation}
	\widehat \tau^{\text{mon}}_{ML} = T^{-1} \cdot \argmin_{t=T+1, \ldots, T_d} \{ R_1(t) + R_2(t) \}  \label{eq:MLestMon}
\end{equation}
to estimate $\tau^*_{\text{mon}} = T^* / T_d$,
where $T_d$ denotes the detection time point, which is the stopping time at which the detector statistic exceeds the boundary function for the first time.
The maximum likelihood estimator is very accurate if the breakpoint is located in the middle of the sample.
However, by construction, the true breakpoint $T^*$ tends to be close to the stopping time $T_d$, and $R_2(T^*)$ is computed from very few observations, which may lead to a large finite sample estimation error for the maximum likelihood estimator.
A theoretical explanation for this effect is given in \cite{casini2018a,casini2021a}, where the finite-sample distribution of the least squares estimator is investigated using a continuous record asymptotic framework.

To bypass this problem, we use backwardly cumulated recursive residuals to estimate the relative break location.
In the single break model, $\|BQ_{\rT,T}\|$ is asymptotically proportional to $\| h(1) - h(r) \|$, which is constant in the pre-break period and decreases to zero in the post-break period.
When scaled by its asymptotic standard deviation, the detector is asymptotically proportional to $\|h(1) - h(r)\|/\sqrt{1-r}$, which in turn (see equation \eqref{eq:h(r)}) is proportional to 
\begin{equation*}
	\big(- \ln(\tau^*) 1_{\{r < \tau^*\}} - \ln(r) 1_{\{ r \geq \tau^* \}}\big)/\sqrt{1-r},	
\end{equation*}
where the maximum is attained at $r = \tau^*$.
Accordingly, we consider
\begin{equation}
	\widehat \tau_{\text{ret}} = \frac{1}{T} \cdot \argmax_{t= 1, \ldots, T} \frac{\|BQ_{t,T} \|}{\sqrt{T-t+1}}, \qquad 
	\widehat \tau_{\text{mon}} = \frac{1}{T} \cdot \argmax_{t=T+1, \ldots, T_d} \frac{\|BQ_{t,T_d} \|}{\sqrt{T_d+t-1}}. \label{eq:BQest}
\end{equation}

\begin{theorem} \label{thm_breakdate}
Let $\beta_t = \beta_0 + \delta 1_{\{t/T \geq \tau^*\}}$, where $\delta \neq 0$, and let Assumption \ref{assumption} hold true. 
If $\tau^* \in (0,1]$, then $\widehat \tau_{\text{ret}} \pplim \tau^*$, as $T \to \infty$;
if $\tau^* \in (1,T_d/T]$, then $\widehat \tau_{\text{mon}} \pplim \tau^*$, as $T \to \infty$.
\end{theorem}
This result implies that the breakpoint estimators \eqref{eq:BQest} are consistent, as $T \to \infty$.

%% file: 7_simulations.tex
\section{Finite sample performance} \label{sec:simulations}
\begin{table}[tb]
	\caption{Empirical sizes and powers of the $R$-tests}
	\centering
\begin{footnotesize}
\begin{tabular}{l|rrrr|rrrr|rrrr}
	& & \multicolumn{2}{c}{model I} & & & \multicolumn{2}{c}{model II} & & & \multicolumn{2}{c}{model III} \\ \hline
& $\mathcal{BQ}_T$ & $\mathcal{SBQ}_T$ & $\mathcal Q_T$ & supW & $\mathcal{BQ}_T$ & $\mathcal{SBQ}_T$ & $\mathcal Q_T$ & supW & $\mathcal{BQ}_T$ & $\mathcal{SBQ}_T$ & $\mathcal Q_T$ & supW  \\  \hline
size & 4.3 & 3.4 & 4.2 & 4.6 & 4.4 & 4.1 & 4.1 & 4.8 & 4.9 & 4.2 & 5.3 & 5.8 \\ \hline
$\tau^*$  & \multicolumn{12}{c}{power}  \\ \hline
0.1 & 54.4 & 67.4 & 73.6 & 51.8 & 51.2 & 63.0 & 68.2 & 53.0 & 24.3 & 26.5 & 40.5 & 36.1 \\ 
0.4 & 99.8 & 99.5 & 92.3 & 99.8 & 99.9 & 99.7 & 93.2 & 99.9 & 98.2 & 94.6 & 66.9 & 98.7 \\ 
0.6 & 99.9 & 99.4 & 72.5 & 99.8 & 99.8 & 99.3 & 72.9 & 99.9 & 99.4 & 96.8 & 47.0 & 98.8 \\ 
0.9 & 56.4 & 25.2 & 5.8 & 46.9 & 52.3 & 33.2 & 6.0 & 49.8 & 49.8 & 19.9 & 6.6 & 34.7 \\ 
\hline
\end{tabular}
\hspace*{0.1ex}
\parbox{15.4cm}{ \vspace{0.5ex} \spacingset{1} \scriptsize
Note: Rejection rates of the retrospective tests are reported for a significance level of $5\%$ and a sample size of $T=200$.
The results are obtained from 100,000 Monte Carlo repetitions under the linear boundaries $d_{\text{lin}}(r)$ and $d_{\text{sbq}}(r,s)$.
The sup-Wald test by \cite{andrews1993} with trimming parameter $0.15$ is denoted as supW.
}
\end{footnotesize}
\label{tab:Rtests}
\end{table}

We illustrate the finite sample performance of the \textit{$R$-tests} and \textit{$M$-tests} for the models
\begin{align}
	y_t &= \gamma_t + u_t, 	\label{eq_simmodel1} \tag{model I} \\
	y_t &= 1 + \gamma_t z_t + u_t,    \tag{model II} \label{eq_simmodel2} \\
	y_t &= \gamma_t + 0.5 y_{t-1} + u_t, \quad (t=1, \ldots, T),   \tag{model III} \label{eq_simmodel3}
\end{align}
where $\gamma_t = 0.8 \cdot 1_{\{t/T \geq \tau^*\}}$, $u_t$ and $e_t$ are independent and $NID(0,1)$, and $z_t = (1-0.5L) e_t$, where $L$ is the lag operator.
For \ref{eq_simmodel1} and \ref{eq_simmodel2} we consider the full structural break tests, and for \ref{eq_simmodel3} partial break tests with $H=(1,0)$ are considered (see Remark \ref{rem:partial}).

\subsection{Retrospective tests (R-tests)}

In Table \ref{tab:Rtests} the empirical sizes and powers of the retrospective tests are compared with that of the sup-Wald test of \cite{andrews1993}.
First, we observe that $\mathcal{BQ}_T$ and $\mathcal{SBQ}_T$ outperform $\mathcal{Q}_T$, except for the case $\tau^* = 0.1$.
Second, while $\mathcal{Q}_T$ has much lower power than the sup-Wald test, the reversed order cumulation structure in $\mathcal{BQ}_T$ and $\mathcal{SBQ}_T$ seems to compensate for the weakness of $\mathcal{Q}_T$. 
\cite{andrews1993} showed that the sup-Wald test is weakly optimal in the sense that, in the case of a single structural break, its asymptotic local power curve approaches the power curve from the infeasible point optimal maximum likelihood test, as the significance level tends to zero.
Within the framework of the considered models, $\mathcal{BQ}_T$ performs similarly well as the sup-Wald test and thus has comparably good power properties as the weakly optimal test.
In contrast to $\mathcal{SBQ}_T$, the the sup-Wald test is not suitable for monitoring since its statistic is not measurable with respect to the filtration of information at time $t$.

\subsection{Monitoring procedures (M-tests)}

\begin{table}[tb]
	\caption{Empirical sizes and mean detection delays of the fixed endpoint $M$-tests}
	\centering
\begin{footnotesize}
\begin{tabular}{l|rr|rr|rr}
& \multicolumn{2}{c|}{model I} & \multicolumn{2}{c|}{model II} & \multicolumn{2}{c}{model III} \\ \hline
 & $\mathcal{SBQ}$ & $\mathcal Q$ & $\mathcal{SBQ}$ & $\mathcal Q$ & $\mathcal{SBQ}$ & $\mathcal Q$ \\ \hline
size & 4.0 & 4.6 & 6.4 & 5.6 & 4.4 & 5.0 \\ \hline
$\tau^*$  & \multicolumn{6}{c}{absolute mean detection delay} \\ \hline
1.1 & 27.4 & 27.3 & 26.9 & 27.1 & 29.1 & 28.8 \\ 
1.2 & 26.6 & 31.9 & 26.2 & 31.1 & 27.9 & 34.6 \\ 
1.3 & 26.1 & 36.5 & 25.9 & 35.3 & 27.2 & 40.2 \\ 
1.4 & 25.9 & 40.9 & 25.7 & 39.3 & 26.8 & 45.1 \\ 
1.5 & 25.8 & 44.7 & 25.4 & 42.2 & 26.4 & 47.7 \\ \hline \hline
\end{tabular}
\hspace*{0.1ex}
\parbox{14cm}{ \vspace{0.5ex} \spacingset{1} \scriptsize
Note: The empirical sizes are reported in percentage points. 
Starting in row 4, the empirical mean detection delays are presented for different breakpoints $\tau^*$.
Results are obtained for a significance level of $5\%$, a pre-monitoring sample size of $T=200$, a monitoring period of $m=2$, boundaries $d_{\text{lin}}(r)$ and $d_{\text{sbq}}(r,s)$, and from 100,000 Monte Carlo repetitions.
Critical values from Tables \ref{tab_crit1} and \ref{tab_crit2} are implemented.
}
\end{footnotesize}
\label{tab:fixeddelay}
\end{table}

Fixed endpoint {\it $M$-tests} are particularly useful when the monitoring period is short. 
{\it $M$-tests} with infinite horizon can be used for long monitoring periods of arbitrary length.
To evaluate the performance of the $M$-tests for finite samples, we simulate the same models as for the {\it $R$-tests} for time points $t=1, \ldots, mT$, where we specify $m = 2$ for the fixed endpoint tests and $m=20$ for the infinite horizon tests.
The results in Table \ref{tab:fixeddelay} show that the mean delay for $\mathcal{SBQ}_{T,m}$ is much lower than that of $\mathcal{Q}_{T,m}$ and is almost constant across the breakpoint locations.

\begin{table}[tb]
	\caption{Empirical sizes and mean detection delays of the infinite horizon $M$-tests}
	\centering
\begin{footnotesize}
\begin{tabular}{l|rrrr|rr|rrrr}
	& & \multicolumn{2}{c}{model I} & &  \multicolumn{2}{c|}{model II}  & & \multicolumn{2}{c}{model III} \\ \hline
 & $\mathcal{SBQ}$ & $\mathcal{Q}$ & CSW & FR & $\mathcal{SBQ}$ & $\mathcal{Q}$ & $\mathcal{SBQ}$ & $\mathcal{Q}$ & CSW & FR  \\  \hline
size & 3.8 & 4.8 & 2.2 & 4.3 & 6.4 & 5.8 & 6.3 & 5.2 & 2.2 & 6.1 \\ \hline
\multicolumn{1}{l|}{$\tau^*$}  & \multicolumn{10}{c}{absolute mean detection delay} \\ \hline
1.5 & 27.6 & 46.4 & 70.2 & 53.6 & 27.3 & 44.2 & 29.7 & 52.8 & 84.1 & 51.8 \\ 
2 & 33.8 & 69.5 & 80.3 & 76.7 & 32.9 & 65.6 & 36.7 & 82.2 & 94.3 & 74.8 \\ 
4 & 57.2 & 162.2 & 111.9 & 169.8 & 52.9 & 150.8 & 64.4 & 200.9 & 126.3 & 167.7 \\ 
6 & 81.1 & 254.9 & 136.9 & 263.6 & 72.3 & 235.4 & 93.9 & 319.6 & 151.3 & 261.6 \\ 
\hline \hline
\end{tabular}
\hspace*{0.1ex}
\parbox{14cm}{ \vspace{0.5ex} \spacingset{1} \scriptsize
Note: The empirical sizes are reported in percentage points. 
Starting in row 4, the empirical mean detection delays are presented for different breakpoints $\tau^*$.
Results are obtained for a significance level of $5\%$, a pre-monitoring sample size of $T=200$, a monitoring period of $m=20$, boundaries $d_{\text{lin}}(r)$ and $d_{\text{inf}}(r,s)$, and from 100,000 Monte Carlo repetitions.
Critical values from Table \ref{tab:crit3} are implemented.
The univariate infinite horizon {\it $M$-tests} by \cite{chu1996} and \cite{fremdt2015} with boundary parameter $0.25$ are denoted as CSW and FR, respectively.
}
\end{footnotesize}
\label{tab:infdelay}
\end{table}

For infinite horizon monitoring, $\mathcal{SBQ}_{T,\infty}$ performs similarly well compared to conventional tests (see Table \ref{tab:infdelay}). 
The detection delay of $\mathcal{Q}_{T,\infty}$ is much higher than that of $\mathcal{SBQ}_{T,\infty}$, and the gap increases further with increasing $\tau^*$. 
Compared to the tests of \cite{chu1996} and \cite{fremdt2015}, we find a similar picture. 
Note that the two alternative tests have no power in \ref{eq_simmodel2} and are therefore omitted for this case.

\subsection{Breakpoint estimators}

\begin{table}[tb]
\centering
\caption{Bias and RMSE of breakpoint estimators}
\begin{footnotesize}
\begin{tabular}{r|rrrr|rrrr}
& \multicolumn{4}{c|}{$T=100$} & \multicolumn{4}{c}{$T=200$} \\ \hline
 & \multicolumn{2}{c}{Bias} & \multicolumn{2}{c|}{RMSE} & \multicolumn{2}{c}{Bias} & \multicolumn{ 2}{c}{RMSE} \\ \hline
$\tau^*$ & BQ & ML & BQ & ML & BQ & ML & BQ & ML \\ 
0.5 & $-$0.03 & $\phantom{-}$0.01 & $\phantom{-}$0.14 & $\phantom{-}$0.11 & $-$0.02 & $\phantom{-}$0.01 & $\phantom{-}$0.08 & $\phantom{-}$0.05 \\ 
0.65 & $-$0.03 & $\phantom{-}$0.00 & $\phantom{-}$0.14 & $\phantom{-}$0.12 & $-$0.02 & $\phantom{-}$0.00 & $\phantom{-}$0.08 & $\phantom{-}$0.05 \\ 
0.8 & $-$0.03 & $-$0.04 & $\phantom{-}$0.15 & $\phantom{-}$0.18 & $-$0.01 & $-$0.01 & $\phantom{-}$0.09 & $\phantom{-}$0.08 \\ 
0.85 & $-$0.04 & $-$0.07 & $\phantom{-}$0.17 & $\phantom{-}$0.22 & $-$0.02 & $-$0.02 & $\phantom{-}$0.10 & $\phantom{-}$0.11 \\ 
0.9 & $-$0.06 & $-$0.13 & $\phantom{-}$0.19 & $\phantom{-}$0.30 & $-$0.03 & $-$0.04 & $\phantom{-}$0.12 & $\phantom{-}$0.17 \\ 
0.95 & $-$0.10 & $-$0.25 & $\phantom{-}$0.24 & $\phantom{-}$0.43 & $-$0.05 & $-$0.14 & $\phantom{-}$0.17 & $\phantom{-}$0.32 \\ 
0.97 & $-$0.13 & $-$0.33 & $\phantom{-}$0.28 & $\phantom{-}$0.50 & $-$0.08 & $-$0.24 & $\phantom{-}$0.22 & $\phantom{-}$0.42 \\ 
0.99 & $-$0.20 & $-$0.44 & $\phantom{-}$0.35 & $\phantom{-}$0.58 & $-$0.15 & $-$0.40 & $\phantom{-}$0.30 & $\phantom{-}$0.56 \\  \hline \hline
\end{tabular}
\hspace*{0.4ex}
\parbox{12cm}{ \vspace{0.5ex} \spacingset{1} \scriptsize
Note: The bias and root mean squared error (RMSE) for the break date estimators \eqref{eq:MLest} and \eqref{eq:BQest} are reported based on 100,000 Monte Carlo repetitions, where model \eqref{eq_simmodel1} is simulated for $t = 1, \ldots, T$.
BQ denotes the backward CUSUM estimator \eqref{eq:BQest}, and ML denotes the maximum likelihood estimator \eqref{eq:MLest}.
}
\end{footnotesize}
\label{tab:breakestimation}
\end{table}

To compare the breakpoint estimator in equation \eqref{eq:BQest} with its maximum likelihood benchmark in \eqref{eq:MLest} and \eqref{eq:MLestMon}, we present Monte Carlo simulation results in Table \ref{tab:breakestimation}.
If the break $\tau^*$ is located after $85 \%$ of the sample, the estimator based on backwardly cumulated recursive residuals has a much lower bias and root mean squared error than the maximum likelihood estimator, which is due to the fact that the post-break entails only few observations.

%% file: 8_covid19.tex
\begin{figure}[t]
\caption{Ljung-Box and robust Q statistics of the residuals}
\centering
\includegraphics[scale=0.39]{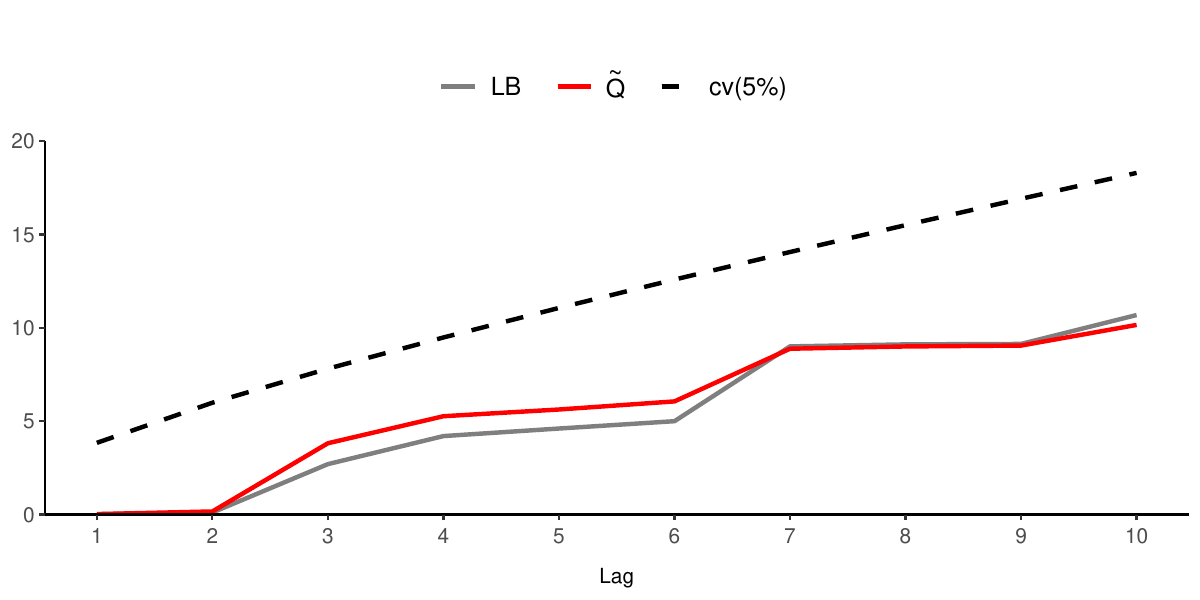} 
\includegraphics[scale=0.39]{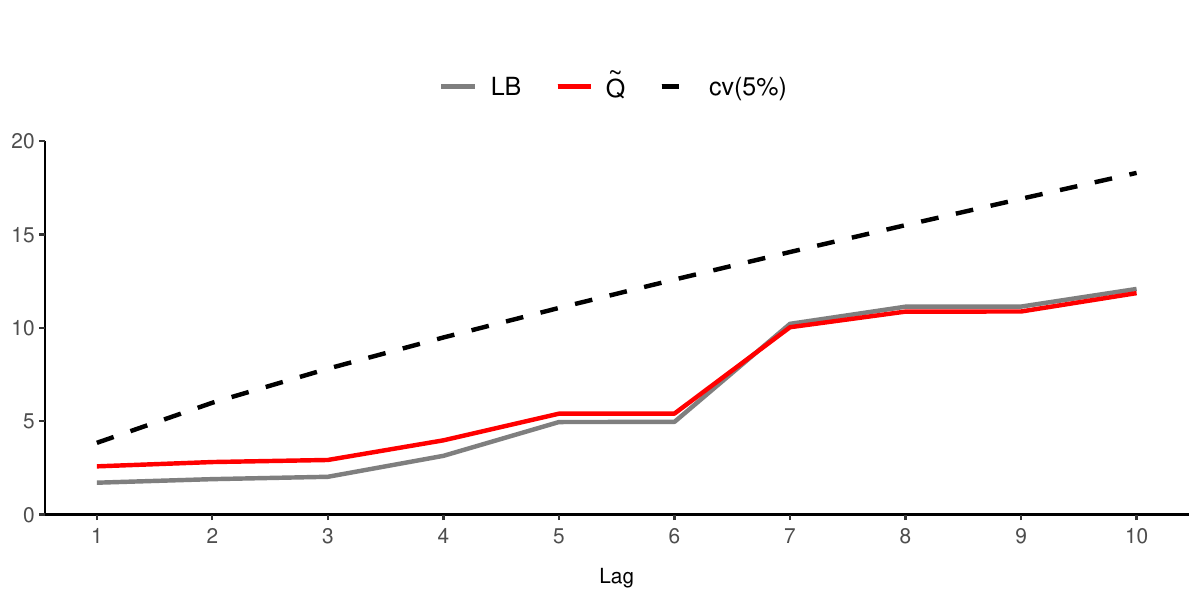}
\parbox{16cm}{ \vspace{0.5ex} \spacingset{1} \scriptsize
Note: The cumulative Ljung-Box (LB) and robust Q-statistics ($\widetilde Q$) of \cite{dalla2020} are plotted for the residuals of model \eqref{eq1}. The left plot shows the statistics using the sample from the first pre-monitoring training period (April 10 until May 21) and the right plot for the second pre-monitoring training period (July 20 until August 30). The dashed line indicates the $5\%$ critical values. The plots are created using the \texttt{R}-package \textit{testcorr} provided by \cite{dalla2020}.
}
\label{fig:Covid2}
\end{figure}

\section{Empirical application to Covid-19 infections} \label{sec:covid}

\begin{figure}[hp]
\caption{Monitoring daily new COVID-19 infections in the US}
\centering
\includegraphics[scale=0.214]{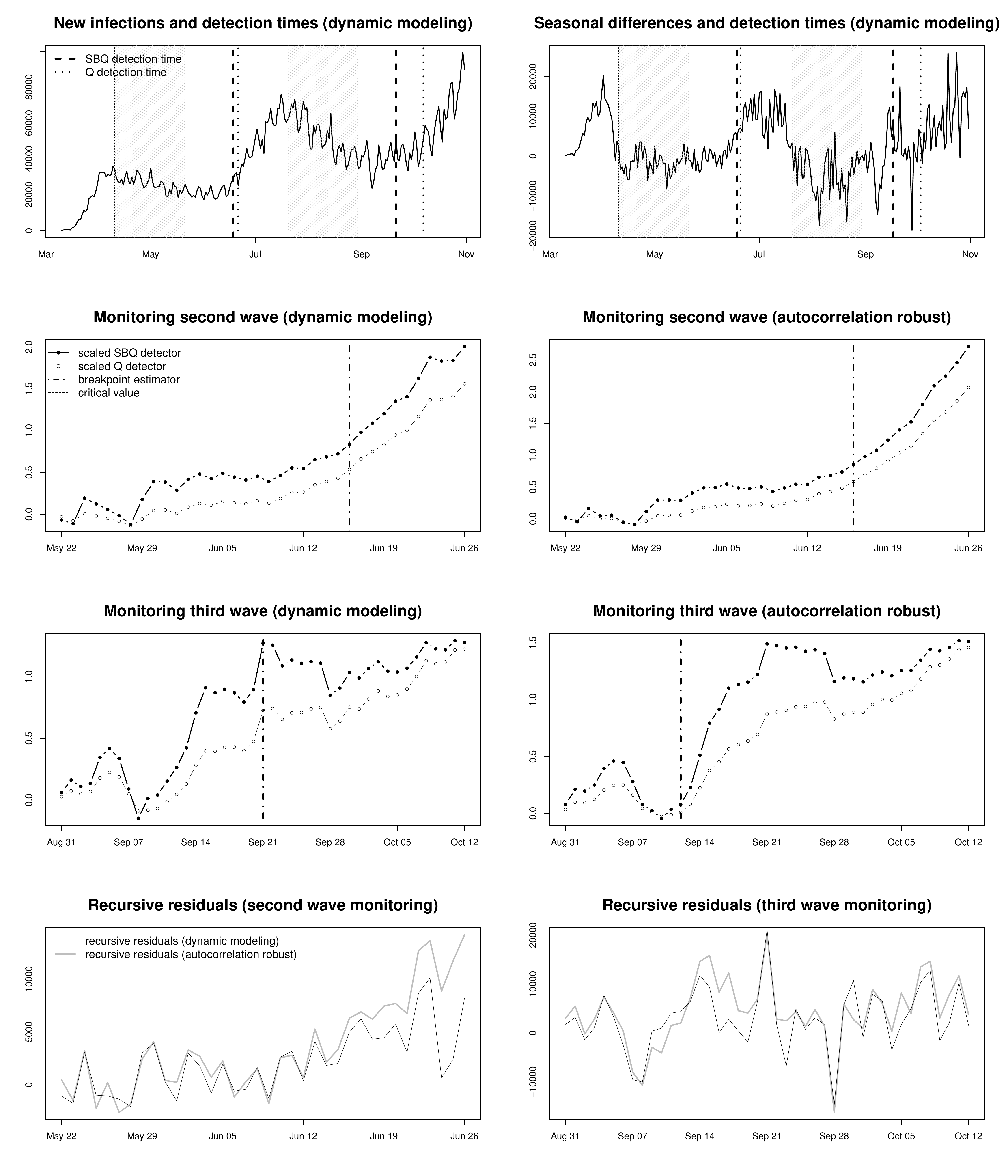}
\parbox{16cm}{ \vspace{0.5ex} \spacingset{1} \scriptsize
Note: The results of the {\it $M$-tests} for the dynamic model (left column) and the autocorrelation robust {\it $M$-tests} (right column) are presented.
The plots show the new infections (first row) and the seasonally differenced infections (second row). The third and fourth row present the standardized recursive residuals (solid), the scaled stacked backward CUSUM detector (dashed), and the scaled forward CUSUM detector (dotted).
The shaded areas represent the pre-monitoring training periods, 
the dashed vertical lines represent the detection time point of the stacked backward CUSUM,
the dotted vertical lines are the detection times of the forward CUSUM, and the vertical dash-dotted lines are the backward CUSUM breakpoint estimators given by equation \eqref{eq:BQest}. 
The horizontal solid line represents the critical boundary of the scaled detectors.
}
\label{fig:Covid}
\end{figure}

We consider the time series $y_t$ of daily new cases of SARS-CoV-2 infections in the US during the first months of the COVID-19 pandemic.
The starts of the first three waves of infections are observed in March, June, and September 2020.
To monitor for a second and third wave of infections, we consider pre-monitoring training samples of six weeks ($T=46$) starting at the times after the first and second peak. 
The training sample periods are given by April 10 until May 21, and July 20 until August 30.
To account for a seasonal unit root in $y_t$ we consider the seasonally differenced series $\widetilde y_t = y_t - y_{t-7}$. 
Weekly differences are used to account for the seasonalities resulting from the weekly reporting pattern of COVID-19 data with lower numbers on weekends.
We estimate the dynamic model
\begin{equation}
\widetilde y_t = \phi_0 + \phi_1 \widetilde y_{t-2} + \phi_2 \widetilde y_{t-7} + u_{t} = x_t^\T \beta + u_{t}, \label{eq1} 
\end{equation}
where $x_t=(1,\widetilde y_{t-1}, \widetilde y_{t-7})^\T$, and $\beta=(\phi_0,\phi_1,\phi_2)^\T$. 
The parameters for lags 2 and 7 are the only significant autoregressive parameters.

Both the Ljung-Box and the robust Q-statistic of \cite{dalla2020} do not indicate any significant autocorrelation in the residuals for the pre-break training periods (see Figure \ref{fig:Covid2}).
We consider the infinite-horizon stacked backward CUSUM statistic for a break in the intercept and the infinite horizon forward CUSUM of \cite{chu1996}.
We are interested in detecting positive changes in the intercept $\phi_0$ and apply one-sided infinite horizon monitoring statistics with a significance level of $5 \%$.
The critical values for the partial right-sided tests are given by those of the full test with $\alpha = 0.1$ and $k=1$ (see Remark \ref{rem:partial}).
We consider the infinite-horizon stacked backward CUSUM statistic for a break in the intercept and the infinite horizon forward CUSUM of \cite{chu1996}.
To compare the detector statistics, we scale them by their boundaries and critical values, 
so that $H_0$ is rejected in favor of a positive change in $\phi_0$ if the detector exceeds unity, respectively.
An alternative to the dynamic modeling in \eqref{eq1} is to apply autocorrelation robust {\it $M$-tests} to the model $\widetilde y_t = \phi_0 + u_t$, where we replace $\widehat{\sigma}_T$ with the long-run variance estimator by \cite{newey1987} (see Remark \ref{rem:serialcorrelation}).

The results are presented in Figure \ref{fig:Covid}.
Both monitoring procedures find an indication for a rise in SARS-CoV-2 infections in the US at the end of June and the end of September.
The stacked backward CUSUM detects the breaks much earlier and becomes significant between 2 and 16 days before the forward CUSUM becomes significant. 
This confirms our theoretical analysis and shows that precious time can be saved by applying the backward monitoring scheme.

%% file: 9_conclusion.tex
\section{Conclusion} \label{sec:Conclusion}

In this paper we propose two alternatives to the conventional CUSUM detectors by \cite{brown1975} and \cite{chu1996}.
It has been demonstrated that cumulating the recursive residuals backwardly results in much higher power than using forwardly cumulated recursive residuals, in particular if the break is located at the end of the sample. 
Accordingly, the backward scheme is especially attractive for on-line monitoring. 
To this end, the stacked triangular array of backwardly cumulated recursive residuals is employed and we find that this approach yields a much lower detection delay than the monitoring procedure by \cite{chu1996}.
Due to the multivariate nature of our tests, they also have power against structural breaks that do not affect the unconditional mean of the dependent variable.
We also propose a new break date estimator which outperforms conventional estimators if the break is located at the end of the sample.

%% file: proofs.tex
\subsection{Auxiliary lemmas}

We first present some auxiliary lemmas which we require for the main proofs.

\begin{lemma} \label{lem_mfclt1}
Let Assumptions \ref{ass:model1} and \ref{ass:md} hold true.
Then, for any fixed $m < \infty$, as $T \to \infty$,
\begin{align*}
	\frac{1}{\sqrt{T}}\sum_{t=1}^{\rT} x_t u_t \Rightarrow \sigma C^{1/2} W^{(k)}(r), \quad r \in [0, m].
\end{align*}
\end{lemma}

\begin{proof}
Note that $M^{-1/2} \sum_{t=1}^{\sM} x_t u_t \Rightarrow \Omega^{1/2} W^{(k)}(s)$, $s \in [0,1]$, 
as $M \to \infty$, by Theorem 7.19 in \cite{white2001}.
Then, on the space $D([0,m])^k$, as $T \to \infty$,
\begin{align}
	\frac{1}{\sqrt{T}}\sum_{t=1}^{\rT}  x_t u_t = \frac{\sqrt m}{\sqrt{M}}\sum_{t=1}^{\lfloor (r/m) M \rfloor  }  x_t u_t 
	\Rightarrow  \sqrt m \Omega^{1/2}  W^{(k)}(r/m) \Deq \Omega^{1/2}  W^{(k)}(r), \quad r \in [0,m]. \label{eq:fclttransform}
\end{align}
\end{proof}

\begin{assumption}
There exists $\kappa \geq 8$ such that $\sup_t E[\| x_t \|^\kappa] < \infty$, $\sup_t E[|u_t|^\kappa] < \infty$, and $(x_t, u_t)'$ is strong mixing of size $-\kappa/(\kappa-6)$.
\label{ass:mixing}
\end{assumption}

\begin{lemma} \label{lem_mfclt2}
Let Assumptions \ref{ass:model1} and \ref{ass:mixing} hold true.
Then, for any fixed $m < \infty$, as $T \to \infty$,
\begin{align*}
	\frac{1}{\sqrt{T}}\sum_{t=1}^{\rT} x_t u_t \Rightarrow \Omega^{1/2} W^{(k)}(r), \quad r \in [0, m].	
\end{align*}
\end{lemma}

\begin{proof}
We have $M^{-1/2} \sum_{t=1}^{\sM} x_t u_t \Rightarrow \Omega^{1/2} W^{(k)}(s)$, $s \in [0,1]$, as $M \to \infty$, by Corollary 4.2 in \cite{wooldridge1988}. The result follows by analogous arguments as in \eqref{eq:fclttransform}.
\end{proof}

\begin{lemma} \label{lem:cusumapprox}
Let $\beta_t = \beta_0$ for all $t \in \mathbb N$, and let Assumption \ref{ass:model1} hold true.
Moreover, let either Assumption \ref{ass:md} or \ref{ass:mixing} hold true.
Then, for any $m < \infty$, as $T \to \infty$,
\begin{align*}
	\max_{1 \leq t \leq mT} \frac{1}{\sqrt T} \bigg\| Y_t - \sum_{j=1}^{t-1} \frac{1}{j} Y_j  - X_t \bigg\| = o_P(1),
\end{align*}
where $X_t = \sum_{j=1}^t x_j w_j$ and $Y_t = \sum_{j=1}^t x_j u_j$. 
\end{lemma}

\begin{proof}
Let $f_j = (1 + (j-1)^{-1} x_j' \widehat C_{j-1}^{-1} x_j )^{-1/2} 1_{\{ j > k \}}$.
Since $\widehat \beta_j = \beta_0 + j^{-1} \widehat C_j^{-1} Y_j$, we can represent the recursive residuals as
\begin{align*}
	w_j = f_j (y_j - x_j' \widehat \beta_{j-1}) 
		= f_j\big(u_j - (j-1)^{-1} x_j' \widehat C_{j-1}^{-1} Y_{j-1}\big).
\end{align*}		
The multivariate cumulative sum can be written as
\begin{align*}
	X_t = \sum_{j=1}^t x_j u_j f_j - \sum_{j=1}^{t-1} \frac{f_{j+1}}{j} x_{j+1} x_{j+1}' \widehat C_{j}^{-1} Y_{j},
\end{align*}
and we have
\begin{align*}
	T^{-1/2} \bigg( Y_t - \sum_{j=1}^{t-1} \frac{1}{j} Y_j  - X_t \bigg) = T^{-1/2} ( A_{t} + B_{t}),
\end{align*}
where
\begin{align}
	A_t =  \sum_{j=1}^t x_j u_j (1 - f_j), \quad
	B_t =  \sum_{j=1}^{t-1}  \big( f_{j+1} x_{j+1} x_{j+1}' \widehat C_{j}^{-1} - I_k \big) \frac{1}{j} Y_j. 	\label{eq:AtBt}
\end{align}
It remains to show that $\max_{1 \leq t \leq mT} T^{-1/2} \| A_t \| = o_P(1)$ and $\max_{1 \leq t \leq mT} T^{-1/2} \| B_t \| = o_P(1)$.
For the first part, note that, for $j > k$,
\begin{align}
	0 \leq \sqrt{j-1} (1 - f_j) = \frac{\sqrt{j-1}\sqrt{1 + \frac{1}{j-1} x_j' \widehat C_{j-1}^{-1} x_j} - 1}{\sqrt{1 + \frac{1}{j-1} x_j' \widehat C_{j-1}^{-1} x_j}}
	\leq \sqrt{x_j' \widehat C_{j-1}^{-1} x_j},	\label{eq:ft-aux}
\end{align}
which implies that $\sqrt j (1 - f_j) = O_P(1)$, as $j \to \infty$.
Since $\widehat C_j$ is uniformly positive definite, there exists a uniformly minimal eigenvalue $\lambda_{min} > 0$, which is defined as the infimum of all eigenvalues of the matrices $\{\widehat C_j\}_{j > k}$.
Any Rayleigh quotient of $C_{j-1}^{-1}$ is bounded above by $\lambda_{min}^{-1} < \infty$. Therefore, for any $j$ and $\delta > 0$,
\begin{align*}
	E[|\sqrt j (1 - f_j)|^\delta] \leq E[|x_j' \widehat C_{j-1}^{-1} x_j|^{\delta/2}] \leq \lambda_{min}^{-\delta/2} E[|x_j'x_j|^{\delta/2}] \leq \lambda_{min}^{-\delta/2} E[\|x_j\|^{\delta}].
\end{align*}
Hence, by the Hölder inequality,
\begin{align}
	&E\big[\|x_j u_j \sqrt j (1-f_j) \|^{2+\epsilon}\big]
	\leq \big(E\big[\|x_j \|^{6+3\epsilon}\big]  E\big[|u_j|^{6+3\epsilon}\big] E\big[|\sqrt j (1-f_j)|^{6+3\epsilon}\big] \big)^{1/3} < \infty,	\label{eq:ft-aux2}
\end{align}
for any $0 < \epsilon \leq (\kappa - 6)/3$.
Then,
\begin{align}
	&Var\Big[ \max_{1 \leq t \leq mT} T^{-1/2} \| A_t \| \Big] 
	\leq \frac{1}{T} Var \Bigg[ \sum_{j=1}^{\mT} \|x_j u_j (1-f_j) \| \Bigg]  = A_{T,1}^* + A_{T,2}^*,	\label{eq:meanconvergence-aux1}
\end{align}
where 
\begin{align}
	A_{T,1}^*  = \frac{1}{T} \sum_{j=1}^{\mT}  Var \big[ \| x_j u_j (1-f_j) \| \big]
	\leq \frac{1}{T} \sum_{j=1}^{\mT} \frac{1}{j} E \big[ \| x_j u_j \sqrt{j} (1-f_j) \|^2 \big], \label{eq:meanconvergence-aux2}
\end{align}
which is $O(\ln(mT) / T)$ since the harmonic series satisfies $\sum_{j=1}^{\mT} j^{-1} = O(\ln(mT))$, and
\begin{align*}
	A_{T,2}^* = \frac{2}{T} \sum_{\tau=1}^{\mT -1} \sum_{j=1}^{\mT - \tau}  Cov \big[\| x_j u_j (1-f_j) \|, \| x_{j+\tau} u_{j+\tau} (1-f_{j+\tau}) \| \big],
\end{align*}
which is zero under Assumption \ref{ass:md}.
Under Assumption \ref{ass:mixing}, we apply Corollary 14.3 of \cite{davidson1994}.
Let $\alpha(\tau)$ be the $\alpha$-mixing sequence of $(x_t',u_t)$, and consider $0 < \epsilon \leq (\kappa - 6)/3$.
Then, by \eqref{eq:ft-aux2}, there exists a constant $K < \infty$, such that
\begin{align}
	A_{T,2}^* &\leq \frac{12}{T} \sum_{\tau=1}^{\mT -1} \sum_{j=1}^{\mT - \tau} \big( E[\| x_j u_j (1-f_j) \|^{2+\epsilon}] E[\|  x_{j+\tau} u_{j+\tau} (1-f_{j+\tau}) \|^{2+\epsilon}] \big)^{1/(2+\epsilon)} \alpha(\tau)^{1-\frac{2}{2 + \epsilon}} \nonumber \\
	&\leq \frac{12 K}{T} \sum_{\tau=1}^\infty \alpha(\tau)^{1-\frac{2}{2+\epsilon}} \sum_{j=1}^{\mT} \frac{1}{j}
	= O(\ln(mT) / T) = o(1), \label{eq:meanconvergence-aux3}
\end{align}
since $1-2/(2+\epsilon) \leq (\kappa-6)/\kappa$.
Consequently, $\max_{1 \leq t \leq mT} \| A_{t,T} \| = o_P(1)$ by Chebyshev's inequality.
For the second term, we consider the decomposition $B_t = B_{t,1} + B_{t,2}$,
where
\begin{align*}
	B_{t,1} = \sum_{j=1}^{t-1} \frac{1}{j} x_{j+1} x_{j+1}' (f_{j+1} \widehat C_{j}^{-1} - C^{-1}) Y_j,  \quad
	B_{t,2} = \sum_{j=1}^{t-1} \frac{1}{j} ( x_{j+1} x_{j+1}' C^{-1} - I_k) Y_j.
\end{align*}
Note that, by \eqref{eq:ft-aux}, $\sqrt j (f_{j+1}^{-1} \widehat C_j - \widehat C_j) = O_P(1)$, as $j \to \infty$.
Thus, by Lemmas \ref{lem_mfclt1} and \ref{lem_mfclt2}, for any $s > 0$,
$\plim_{T \to \infty} x_{\sT+1} x_{\sT+1}' (f_{\sT+1} \widehat C_{\sT}^{-1} - C^{-1}) T^{-1/2} Y_{\sT} = 0$.
Consequently, as $T \to \infty$,
\begin{align*}
	\max_{1 \leq t \leq mT} T^{-1/2} \|B_{t,1} \| = \sup_{r \in [0,m]} \int_0^r x_{\sT+1} x_{\sT+1}' (f_{\sT+1} \widehat C_{\sT}^{-1} - C^{-1}) T^{-1/2} Y_{\sT} \dd s + o_P(1),
\end{align*}
which is $o_P(1)$ by the continuous mapping theorem.
For $B_{t,2}$, we apply Abel's formula of summation by parts, which is given by
\begin{align}
	\sum_{j=1}^t  A_j  b_j = \sum_{j=1}^t  A_j  b_t + \sum_{j=1}^{t-1} \sum_{i=1}^j  A_i (  b_j -  b_{j+1}), \qquad  A_j \in \mathbb{R}^{k \times k}, \quad  b_j \in \mathbb{R}^k, \quad t \in \mathbb{N}.	\label{eq:Abel}
\end{align}
By setting $A_j = x_{j+1} x_{j+1}' C^{-1} - I_k$ and $b_j = j^{-1} T^{-1/2} Y_j$ we get
\begin{align*}
	T^{-1/2} B_{t,2} &= \sum_{j=1}^{t-1}  A_j  b_{t-1} + \sum_{j=1}^{t-2} \sum_{i=1}^j  A_i (  b_j -  b_{j+1}) 
	= B_{t,T,1}^* + B_{t,T,2}^* + B_{t,T,3}^*,
\end{align*}
where
\begin{align*}
	B_{t,T,1}^* = (\widetilde C_t C^{-1} - I_k) T^{-1/2} Y_{t-1}, \quad 
	B_{t,T,2}^* = \sum_{j=1}^{t-2}  \frac{j-1}{j(j+1)} (\widetilde C_j C^{-1} - I_k) T^{-1/2} Y_j,
\end{align*}
$B_{t,T,3}^* = - \sum_{j=1}^{t-2} \frac{j-1}{j+1} (\widetilde C_j C^{-1} - I_k) T^{-1/2} x_{j+1} u_{j+1}$, and $\widetilde C_t = (t-1)^{-1} \sum_{j=2}^t x_j x_j'$.
Lemmas \ref{lem_mfclt1} and \ref{lem_mfclt2}, the continuous mapping theorem, and the fact that $\plim_{t \to \infty} \widetilde C_t = C$ imply
\begin{align*}
	\max_{1 \leq t \leq mT} \| B_{t,T,1}^* \| = \sup_{r \in [0,m]} \| (\widetilde C_{\rT} C^{-1} - I_k) T^{-1/2} Y_{\rT} \| = o_P(1),
\end{align*}
and
\begin{align*}
	\max_{1 \leq t \leq mT} \| B_{t,T,2}^* \| = \sup_{r \in [0,m]} \int_0^r \|  (\widetilde C_{\rT} C^{-1} - I_k) T^{-1/2} Y_{\rT} \| \dd r = o_P(1).
\end{align*}
For the last term, note that $(\widetilde C_j C^{-1} - I_k) T^{-1/2} x_{j+1} u_{j+1}$ is either a martingale difference sequence or strong mixing. Then, by the fact that we have bounded eighth moments, and by Theorems 24.3 and 24.6 in \cite{davidson1994}, we have $\max_{1 \leq t \leq mT} \| B_{t,T,3}^* \| = o_P(1)$, and the assertion is shown.
\end{proof}

\begin{lemma} \label{lem:cusumapprox2}
Let Assumptions \ref{ass:model1} and \ref{ass:strongprinciple} hold true, and let $\beta_t = \beta_0$ for all $t \in \mathbb N$.
Moreover, let either Assumption \ref{ass:md} or \ref{ass:mixing} hold true.
Then, as $T \to \infty$,
\begin{align*}
	\sup_{t > T} \frac{1}{\sqrt t} \bigg\| Y_t - \sum_{j=1}^{t-1} \frac{1}{j} Y_j  - X_t \bigg\| = o_P(1),
\end{align*}
where $X_t = \sum_{j=1}^t x_j w_j$ and $Y_t = \sum_{j=1}^t x_j u_j$. 
\end{lemma}

\begin{proof} 

Analogously to the proof of Lemma \ref{lem:cusumapprox}, it remains to show that
\begin{align*}
	\sup_{t > T} \frac{1}{\sqrt t} \bigg\| Y_t - \sum_{j=1}^{t-1} \frac{1}{j} Y_j  - X_t \bigg\| =  \sup_{t > T} \frac{1}{\sqrt t} \big\| A_t + B_t \big\| = o_P(1),
\end{align*}	
where $A_t$ and $B_t$ are defined in \eqref{eq:AtBt}, which satisfy $\sup_{t \leq T} T^{-1/2} \| A_t + B_t \| $.
Thus, for any $n \in \mathbb N$,
\begin{align*}
	\sup_{T < t \leq nT} t^{-1/2} \|A_t + B_t \|
	&\leq \sup_{T < t \leq nT} T^{-1/2} \|A_t + B_t \|
	\leq \sup_{t \leq nT} T^{-1/2} \|A_t + B_t \| = o_P(1), \\
	\plim_{t \to \infty} t^{-1/2} \|A_t + B_t \| 
	&\leq \plim_{t \to \infty} \sup_{l \leq t} t^{-1/2} \|A_l + B_l \| = 0, 
\end{align*}
and the assertion follows.
\end{proof}

\begin{lemma} \label{lem:FCLT-rec}
Let Assumption \ref{ass:model1} hold true, let $ \beta_t =  \beta_0$ for all $t$, and let $m < \infty$.
Under Assumption \ref{ass:md},
\begin{align*}
	\frac{1}{\sqrt T} \sum_{t=1}^{\rT}  x_t w_t \Rightarrow \sigma  C^{1/2} W^{(k)}(r), \quad r \in [0, m],
\end{align*}
and, under Assumption \ref{ass:mixing},
\begin{align*}
	\frac{1}{\sqrt T} \sum_{t=1}^{\rT}  x_t w_t \Rightarrow \Omega^{1/2} W^{(k)}(r), \quad r \in [0, m],
\end{align*}
where $W^{(k)}(r)$ is a $k$-dimensional standard Brownian motion.
\end{lemma}

\begin{proof}
Lemma \ref{lem:cusumapprox} implies that
\begin{align}
	\sup_{r \in [0,m]} \frac{1}{\sqrt T} \Big\| Y_{\rT} - \int_0^r z^{-1} Y_{\zT} \dd z - X_{\rT} \Big\| = o_P(1), \label{eq:cusumapprox-aux}
\end{align}
where $X_t = \sum_{j=1}^t x_j w_j$ and $Y_t = \sum_{j=1}^t x_j u_j$. 
Under Assumption \ref{ass:md}, Lemma \ref{lem_mfclt1}, the continuous mapping theorem, and \eqref{eq:cusumapprox-aux} imply
\begin{align*}
	\frac{1}{\sqrt T} X_{\rT} \Rightarrow \sigma C^{1/2}  \bigg( W^{(k)}(r) - \int_0^r z^{-1} W^{(k)}(z) \dd z \bigg) , \quad r \in [0, m].
\end{align*}
Analogously, under Assumption \ref{ass:mixing},
\begin{align*}
	\frac{1}{\sqrt T} X_{\rT} \Rightarrow \Omega^{1/2} \bigg( W^{(k)}(r) - \int_0^r z^{-1} W^{(k)}(z) \dd z \bigg) , \quad r \in [0, m],
\end{align*}
by Lemma \ref{lem_mfclt2}.
It remains to show that, for any $r \geq 0$,
\begin{align}
	W^{(k)}(r) - \int_0^r z^{-1} W^{(k)}(z) \dd z \Deq W^{(k)}(r).	 \label{eq:BMintegral}
\end{align}
Let $W_j(r)$ and $B_j(r)$ be the $j$-th component of $W^{(k)}(r)$ and $B^{(k)}(r)$, respectively.
We show the identities for each $j=1,\ldots,k$, separately.
Using Cauchy-Schwarz and Jensen's inequalities, we obtain $\int_0^r z^{-1} E[|W_j(z)|] \dd z < \infty$ as well as $\int_0^r z^{-1} E[|W_j(r) W_j(z)|] \dd z < \infty$, which justifies the application of Fubini's theorem in the subsequent steps.
Since both $W_j(r)$ and $F(W_j(r)) = W_j(r) - \int_0^r z^{-1} W_j(z) \dd z$ are Gaussian with zero mean, it remains to show that their covariance functions coincide.
Let w.l.o.g.\ $r \leq s$. Then,
\begin{align*}
	&E[F(W_j(r))F(W_j(s))] - E[W_j(r)W_j(s)] \\
	&=\int_0^r \int_0^s\frac{E[W_j(z_1) W_j(z_2)]}{z_1 z_2} \dd z_2 \dd z_1 - \int_0^s \frac{E[W_j(r)W_j(z_2)]}{z_2} \dd z_2  - \int_0^r \frac{E[W_j(s)W_j(z_1)]}{z_1} \dd z_1  \\
	&= (2r + r \ln(s) - r \ln(r)) - (r + r \ln(s) - r \ln(r)) - r = 0,
\end{align*}
and the assertion follows.
\end{proof}

\begin{lemma} \label{lem_SumConvergence}
Let $ h$ be a $\mathbb{R}^k$-valued function of bounded variation, and let $\{ A_t\}_{t \in \mathbb{N}}$ be a sequence of random $(k \times k)$ matrices with $\sup_{r \in [0, m]} \| T^{-1} \sum_{t=1}^{\rT} ( A_t -  A) \|_M = o_P(1)$, where $m < \infty$. Then, as $T \to \infty$,
\begin{align*}
	\sup_{r \in [0, m]} \Big\| \frac{1}{T} \sum_{t=1}^{\rT} ( A_t -  A)  h(\tfrac{t}{T}) \Big\| = o_P(1).
\end{align*}
\end{lemma}

\begin{proof}
 By the application of Abel's formula of summation by parts, which is given in \eqref{eq:Abel}, it follows that
\begin{align*}
	\sum_{t=1}^{\rT} ( A_t -  A)  h(\tfrac{t}{T}) 
	= \sum_{t=1}^{\rT} ( A_t -  A)  h(\tfrac{\rT}{T}) + \sum_{t=1}^{\rT - 1} \sum_{j=1}^t ( A_j -  A) ( h(\tfrac{t}{T}) -  h(\tfrac{t+1}{T})).
\end{align*}
The fact that $ h(r)$ is of bounded variation yields $\sup_{r \in [0, m]} \| h(r) \| = O(1)$ as well as $\sup_{r \in [0, m]} \| \sum_{t=1}^{\rT - 1} ( h(\tfrac{t}{T}) -  h(\tfrac{t+1}{T})) t/T \| = O(1)$.
Consequently, 
\begin{align*}
	\sup_{r \in [0, m]} \Big\| \frac{1}{T} \sum_{t=1}^{\rT} ( A_t -  A)  h(\tfrac{\rT}{T}) \Big\|
	\leq \sup_{r \in [0, m]} \Big\| \frac{1}{T} \sum_{t=1}^{\rT} ( A_t -  A)  \Big\|_M  \Big\|  h(\tfrac{\rT}{T})  \Big\| = o_P(1)
\end{align*}
and
\begin{align*}
	&\sup_{r \in [0, m]} \Big\| \frac{1}{T} \sum_{t=1}^{\rT - 1} \sum_{j=1}^t ( A_j -  A) ( h(\tfrac{t}{T}) -  h(\tfrac{t+1}{T})) \Big\| \\
	&\leq \sup_{r \in [0, m]} \sum_{t=1}^{\rT - 1} \frac{t}{T} \Big\| \frac{1}{t} \sum_{j=1}^t ( A_j -  A) \Big\|_M \Big\| h(\tfrac{t}{T}) -  h(\tfrac{t+1}{T}) \Big\|
	= o_P(1).
\end{align*}
Then, by the triangle inequality, the assertion follows.
\end{proof}

\subsection{Main proofs}

\subsubsection*{Proof of Theorem \ref{thm:fclt}}

Consider the auxiliary sequence $y_t^* = x_t' \beta_0 + u_t$, which coincide with $y_t$ if $\beta_t = \beta_0$. Moreover, define $\widehat \beta_{t-1}^* = ( \sum_{j=1}^{t-1}  x_j  x_j' )^{-1} \sum_{j=1}^{t-1}  x_j y_j^*$ and $f_t = (1 + (t-1)^{-1} x_t' \widehat C_{t-1}^{-1} x_t )^{-1/2} 1_{\{ t > k \}}$. Then, $w_t^* = f_t (y_t^* -  x_t' \widehat{ \beta}_{t-1}^* )$ are recursive residuals from a regression without any structural break in the coefficients.
If $\beta_t = \beta_0 + T^{-1/2} g(t/T)$, we have $y_t =  x_t'  \beta_t + u_t = y_t^* + T^{-1/2}  x_t' g(t/T)$, and 
$\widehat \beta_{t-1} = \widehat \beta_{t-1}^* + T^{-1/2} (t-1)^{-1} \widehat C_{t-1}^{-1} \sum_{j=1}^{t-1}  x_j  x_j'  g(j/T)$, which implies that
\begin{equation*}
	w_t = w_t^* + f_t T^{-1/2}  x_t'  g(t/T) - f_t T^{-1/2} (t-1)^{-1} \widehat C_{t-1}^{-1} \sum_{j=1}^{t-1}  x_j  x_j'  g(j/T).
\end{equation*}
We decompose $T^{-1/2} \sum_{t=1}^{\rT}  x_t w_t =  \mathcal S_{1,T}(r) + \mathcal S_{2,T}(r) + \mathcal S_{3,T}(r)$, where
\begin{align*}
	\mathcal S_{1,T}(r) &= \frac{1}{\sqrt{T}} \sum_{t=1}^{\rT}  x_t w_t^*, \quad
	\mathcal S_{2,T}(r) = \frac{1}{T} \sum_{t=1}^{\rT} f_t  x_t  x_t'  g(\tfrac{t}{T}),  \\
	\mathcal S_{3,T}(r) &= - \frac{1}{T} \sum_{t=1}^{\rT} f_t (t-1)^{-1}  x_t  x_t' \widehat C_{t-1}^{-1} \sum_{j=1}^{t-1}  x_j  x_j'  g(\tfrac{j}{T}). 
\end{align*}
Lemma \ref{lem:FCLT-rec} yields $\mathcal S_{1,T}(r) \Rightarrow \Omega^{1/2} W^{(k)}(r)$, where $\Omega^{1/2} = \sigma  C^{1/2}$ under Assumption \ref{ass:md}.
Analogously to \eqref{eq:meanconvergence-aux1}--\eqref{eq:meanconvergence-aux3}, we have
\begin{align*}
	E \bigg[ \sup_{r \in [0,m]} \Big\| \frac{1}{T} \sum_{t=1}^{\rT} (f_t - 1) x_t x_t' \Big\|_M^2 \bigg] =o(1).
\end{align*}
Moreover, $\sup_{r \in [0,m]} \| \widehat C_{\rT} - C \|_M = o_P(1)$ due to Assumption \ref{ass:model1}(b), and, consequently,
\begin{align}
\sup_{r \in [0,m]} \Big\| \frac{1}{T} \sum_{t=1}^{\rT} f_t x_t x_t' - C \Big\|_M = o_P(1). \label{f_A2convergence}
\end{align}
Since $ g(r)$ is piecewise constant and therefore of bounded variation, Lemma \ref{lem_SumConvergence} yields
\begin{align*}
	&\sup_{r \in [0, m]} \Big\| \mathcal S_2(r) - \int_0^r  C  g(s) \dd s \Big\| 	 
	= \sup_{r \in [0, m]} \Big\| \frac{1}{T} \sum_{t=1}^{\rT} (f_t  x_t  x_t' -  C)  g(\tfrac{t}{T}) \Big\|  
	= o_P(1). 
\end{align*}
For the third term, let
\begin{align*}
	p_{1,t} =  \widehat C_t^{-1} \frac{1}{t} \sum_{j=1}^t x_j x_j' g(\tfrac{j}{T}), \quad
	p_{2,t} = \widehat C_t^{-1} \frac{1}{t} \sum_{j=1}^{t}  C  g(\tfrac{j}{T}), \quad
	p_{3,t} = \frac{1}{t} \sum_{j=1}^{t} g(\tfrac{j}{T}).
\end{align*}
From Assumption \ref{ass:model1}(b) and Lemma \ref{lem_SumConvergence},
$\|p_{1,T} - p_{3,T} \| \leq \|p_{1,T} - p_{2,T} \| + \|p_{2,T} - p_{3,T} \| = o_P(1)$.
Consequently,
\begin{align*}
	&\sup_{r \in [0, m]} \Big\| \mathcal S_{3,T}(r) + \frac{1}{T} \sum_{t=1}^{\rT} f_t  x_t  x_t'  p_{3,t-1} \Big\|
	\leq \int_0^m \| f_{\sT} x_{\sT} x_{\sT}' (p_{3,\sT} - p_{1,\sT}) \dd s \| = o_P(1).
\end{align*}
Since $p_3(r) = p_{3,\rT}$ is a partial sum of a piecewise constant function, it is of bounded variation, and, together with \eqref{f_A2convergence}, Lemma \ref{lem_SumConvergence} yields
\begin{align*}
	\sup_{r \in [0, m]} \Big\| \frac{1}{T} \sum_{t=1}^{\rT} ( f_t^{-1}  x_t  x_t' -  C)  p_{3,t-1} \Big\| = o_P(1),
\end{align*}
which implies that $\sup_{r \in [0, m]} \| \mathcal S_{3,T}(r) + \frac{1}{T}  C \sum_{t=1}^{\rT}  p_{3,t-1} \| = o_P(1)  $, and 
\begin{align*}
	\sup_{r \in [0, m]} \Big\| \mathcal S_{3,T}(r) + \int_0^r \int_0^s \frac{1}{s}  C  g(v) \dd v \dd s \Big\| 
	= o_P(1).
\end{align*}
Consequently, $\sup_{r \in [0, m]} \| \mathcal S_{2,T}(r) + \mathcal S_{3,T}(r) - \sigma C h(t/T) \| = o_P(1)$, and, by Slutsky's theorem,
\begin{align}
	\mathcal S_{1,T}(r) + \mathcal S_{2,T}(r) + \mathcal S_{3,T}(r) \Rightarrow  \sigma  C^{1/2} W^{(k)}(r) + \sigma  C h(r).		\label{eq:thm1-convergenceaux1}
\end{align}
Finally,
\begin{align*}
	Q_T(r) = \widehat \sigma_T^{-1}  C_T^{-1/2} (\mathcal S_{1,T}(r) + \mathcal S_{2,T}(r) + \mathcal S_{3,T}(r)) \Rightarrow W^{(k)}(r) +  C^{1/2}  h(r),
\end{align*}
since $\widehat \sigma_T^2$ is consistent for $\sigma^2$ (see \citealt{kramer1988}).

Note that, if we replace Assumption \ref{ass:md} by Assumption \ref{ass:mixing}, the corresponding limiting result for the autocorrelation robust statistic $\widetilde Q_T(r)$ under $H_0$ follows analogously. 
Moreover, under the local alternatives we have 
\begin{align*}
	\mathcal S_{1,T}(r) + \mathcal S_{2,T}(r) + \mathcal S_{3,T}(r) \Rightarrow  \Omega^{1/2} W^{(k)}(r) + \sigma  C h(r),
\end{align*}
in equation \eqref{eq:thm1-convergenceaux1} so that $\widetilde Q_T(r) \Rightarrow W^{(k)}(r) +  \sigma \Omega^{-1/2} C h(r)$.

\subsubsection*{Proof of Theorem \ref{thm:strongprinciple}}

By Assumption \ref{ass:strongprinciple} there exists a $k$-dimensional standard Brownian motion $W^{(k)}(t)$ such that
\begin{equation*}
	\sup_{t > T} t^{-1/2} \big\| Y_t - \Omega^{1/2} W^{(k)}(t) \big\| = o_P(1),	
\end{equation*}
and there exists a random variable $\xi$ and some $\epsilon > 0$ such that $\| Y_t - \Omega^{1/2} W^{(k)}(t) \| \leq \xi t^{1/2 - \epsilon}$, which implies that
\begin{equation*}
	\sup_{t > T} t^{-1/2} \bigg\| \sum_{j=1}^{t-1} j^{-1} \big(Y_j - \Omega^{1/2} W^{(k)}(j) \big) \bigg\| 
	\leq \sup_{t > T} t^{-\epsilon} \xi \sum_{j=1}^t j^{-1}  = o_P(1).
\end{equation*}
Therefore,
\begin{align*}
	\sup_{t > T} t^{-1/2} \bigg\| Y_t - \sum_{j=1}^{t-1} j^{-1} Y_j  - \Omega^{1/2} \Big(  W^{(k)}(t) - \sum_{j=1}^{t-1} j^{-1} W^{(k)}(j) \Big) \bigg\| = o_P(1).
\end{align*}
By Lemma \ref{lem:cusumapprox2}, from the fact that $T^{-1/2} W^{(k)}(t) \Deq W^{(k)}(t/T)$, and from \eqref{eq:BMintegral}, it follows that there exists another $k$-di\-men\-sion\-al standard Brownian motion $\widetilde W^{(k)}(t)$, such that
\begin{align*}
	\sup_{r > 1} r^{-1/2} \big\| T^{-1/2} X_{\rT} - \Omega^{1/2} \widetilde W^{(k)}(r) \big\| = o_P(1),
\end{align*}
and the assertion follows, since $\Omega = \sigma^2 C$ under Assumption \ref{ass:md}, and since $\widehat \sigma^2_T$, $\widehat C_T$, and $\widehat \Omega_T$ are consistent estimators for their population counterparts.

\subsubsection*{Proof of Theorem \ref{thm:infinite}}

By the triangle inequality, we have
\begin{align*}
	\|  Q_T(r) -  Q_T(s)\| - \| W^{(k)}(r) - W^{(k)}(s)\| \leq \| Q_T(r) - W^{(k)}(r) \| + \|  Q_T(s)  -  W^{(k)}(s) \|
\end{align*}
for any $r$ and $s$.
Thus,
\begin{align*}
	&\sup_{r > 1} \frac{\|  Q_T(r) -  Q_T(1)\|}{d(r-1)} - \sup_{r > 1} \frac{\| W^{(k)}(r) - W^{(k)}(1)\|}{d(r-1)} \\
	&\leq \sup_{r > 1} \frac{\| Q_T(r) - W^{(k)}(r) \|}{d(r-1)}  + \sup_{r > 1} \frac{\|  Q_T(1) -  W^{(k)}(1) \|}{d(r-1)} \\
	&\leq \sup_{r > 1} \bigg( \frac{\| Q_T(r) - W^{(k)}(r) \|}{\sqrt r}  \frac{\sqrt r}{d(r-1)} \bigg) + \sup_{r > 1} \bigg( \frac{\|  Q_T(1) -  W^{(k)}(1) \|}{\sqrt r} \frac{\sqrt r}{d(r-1)} \bigg) \\
	&\leq 2  \bigg( \sup_{r > 1}  \frac{\| Q_T(r) - W^{(k)}(r) \|}{\sqrt r} \bigg)  \bigg( \sup_{r > 1} \frac{\sqrt r}{d(r-1)} \bigg)   = o_P(1)
\end{align*}
for some $k$-dimensional standard Brownian motion $W^{(k)}(r)$, which implies that
\begin{align*}
	\mathcal{Q}_{T, \infty} = \sup_{r \in (1, \infty)} \frac{\|  Q_T(r) -  Q_T(1) \|}{d(r-1)} \Dlim \sup_{r \in (1,\infty)} \frac{\| W^{(k)}(r) - W^{(k)}(1)\|}{d(r-1)}.
\end{align*}
We transform the expression into a supremum over the unit interval.
A $k$-dimensional Brownian motion $W^{(k)}(r)$ and a $k$-dimensional Brownian Bridge $B^{(k)}(r)$ have the distributional relation $B^{(k)}(r) \Deq (1-r) W^{(k)}(r/(1-r))$, which can be verified by comparing their covariance functions.
Consider the bijective function $f: (0, 1) \to (0, \infty)$ given by $f(\eta) = \eta / (1-\eta)$.
The limiting distribution of $\mathcal Q_{T,m}$ can be rearranged as
\begin{align*}
	&\sup_{r \in (1,\infty)} \frac{\| W^{(k)}(r) - W^{(k)}(1)\|}{d(r-1)}
	\Deq \sup_{r \in (0,\infty)} \frac{\| W^{(k)}(r)\|}{d(r)} \\
	&=  \sup_{ \eta \in (0, 1)} \frac{\|W^{(k)}(f(\eta))\|}{d(f(\eta))}
	\Deq \sup_{ \eta \in (0, 1)} \frac{\|B^{(k)}(\eta)\|}{(1-\eta)d\big(\frac{\eta}{1-\eta}\big)}.
\end{align*}

\noindent
For the second result, we analogously have
\begin{align*}
&\sup_{r \in (1,\infty)} \sup_{s \in (1,r)} \frac{\|  Q_T(r) -  Q_T(s)\|}{d(r,s)} - \sup_{r \in (1,\infty)} \sup_{s \in (1,r)} \frac{\| W^{(k)}(r) - W^{(k)}(s)\|}{d(r,s)} \\
	&\leq \sup_{r \in (1,\infty)} \sup_{s \in (1,r)} \frac{\| Q_T(r) - W^{(k)}(r) \|}{d(r,s)} + \sup_{r \in (1,\infty)} \sup_{s \in (1,r)} \frac{\|  Q_T(s)  -  W^{(k)}(s) \|}{d(r,s)} \\
	&\leq 2 \bigg( \sup_{r \in (1,\infty)} \frac{\| Q_T(r) - W^{(k)}(r) \|}{\sqrt r} \bigg) \bigg( \sup_{r \in (1,\infty)} \sup_{s \in (1,r)} \frac{\sqrt r}{d(r,s)} \bigg) = o_P(1)
\end{align*}
for some $k$-dimensional standard Brownian motion $W^{(k)}(r)$.
Then, 
\begin{align*}
	\mathcal{SBQ}_{T, \infty} 
	= \sup_{r \in (1,\infty)} \sup_{s \in (1,r)} \frac{\|  Q_T(r) -  Q_T(s)\|}{d(r,s)}
	\Dlim \sup_{r \in (1,\infty)} \sup_{s \in (1,r)} \frac{\| W^{(k)}(r) - W^{(k)}(s)\|}{d(r,s)}.
\end{align*}
Consider again the bijective function $f$ from above, which analogously implies that
\begin{align*}
	&\sup_{r \in (1, \infty)} \sup_{s \in (1,r)} \frac{\|W^{(k)}(r) - W^{(k)}(s)\|}{d(r,s)} 
	\Deq \sup_{r \in (0, \infty)} \sup_{s \in (0,r)} \frac{\|W^{(k)}(r) - W^{(k)}(s)\|}{d(r+1,s+1)} \\
	&= \sup_{\eta \in (0, 1)} \sup_{s \in (0,f(\eta))} \frac{\|W^{(k)}(f(\eta)) - W^{(k)}(s)\|}{d(f(\eta)+1,s+1)}  
	= \sup_{\eta \in (0, 1)} \sup_{\zeta \in (0,\eta)} \frac{\|W^{(k)}(f(\eta)) - W^{(k)}(f(\zeta))\|}{d(f(\eta)+1,f(\zeta)+1)} \\
	&\Deq \sup_{\eta \in (0, 1)} \sup_{\zeta \in (0,\eta)} \frac{\|B^{(k)}(\eta)/(1-\eta) - B^{(k)}(\zeta)/(1-\zeta)\|}{d( (1 - \eta)^{-1},  (1 - \zeta)^{-1})} \\
	&= \sup_{\eta \in (0,1)} \sup_{\zeta \in (0,r)}\frac{\|(1-\zeta) B^{(k)}(\eta) - (1-\eta) B^{(k)}(\zeta)\|}{(1-\eta)(1-\zeta)d( (1 - \eta)^{-1}, (1 - \zeta)^{-1}) }.
\end{align*}

\subsubsection*{Proof of Theorem \ref{thm_breakdate}}

Adopting the notation of the local break in Theorem \ref{thm:fclt}, we have $\beta_t + T^{-1/2} g(t/T)$ with $ g(t/T) = T^{1/2}  \delta 1_{\{ t \geq T^* \}}$.
Note that
\begin{align*}
	\int_0^r 1_{\{z \geq \tau^*\}} \dd z - \int_0^r \int_0^z \frac{1}{z} 1_{\{v \geq \tau^*\}} \dd v \dd z
	= \int_0^r \frac{\tau^*}{z} 1_{\{z \geq \tau^*\}} \dd z
	= \tau^* \big( \ln(r) - \ln(\tau^*) \big) 1_{\{r \geq \tau^*\}}.
\end{align*}
By equation \eqref{f_h(r)} and Theorem \ref{thm:fclt},
\begin{align*}
	Q_T(r) - T^{1/2} \sigma^{-1} \tau^* C^{1/2} \delta \big( \ln(r) - \ln(\tau^*) \big) 1_{\{r \geq \tau^*\}} \Rightarrow W^{(k)}(r), \quad r \in [0,m].
\end{align*}
Then, by the continuous mapping theorem,
\begin{align*}
	\widehat \tau_{\text{ret}} &= \frac{1}{T} \bigg( \argmax_{1 \leq t \leq T}  \frac{ \| Q_T(1) -  Q_T (\tfrac{t+1}{T}) \| }{\sqrt{T-t+1}} \bigg)
	= \argsup_{r \in (0,1)} \frac{ \| Q_T(1) -  Q_T (r) \| }{\sqrt T \sqrt{1-r}} \\
	&= \argsup_{r \in (0,1)} \frac{ \big( \ln(1) - \ln(\tau^*) \big)1_{\{ 1\geq\tau^* \}} - \big( \ln(r) - \ln(\tau^*) \big) 1_{\{ r \geq \tau^* \}} }{\sqrt{1-r}} + o_P(1) \\
	&= \tau^* + o_P(1),
\end{align*}
since $-\ln(\tau^*)/\sqrt{1-r}$ is strictly increasing for $r \in (0, \tau^*)$ and $-\ln(r)/\sqrt{1-r}$ is strictly decreasing for $r \in [\tau^*, 1)$.
Analogously, if $\tau^* \in (1, \tau_d]$, 
\begin{align*}
	\widehat \tau_{\text{mon}} &= \frac{1}{T} \bigg( \argmax_{T < t \leq T_d}  \frac{ \| Q_{T_d}(1) -  Q_{T_d} (\tfrac{t+1}{T_d}) \| }{\sqrt{T_d-t+1}}  \bigg)
	= \argsup_{r \in (0,1)} \frac{ - \ln(r) 1_{\{ r \geq \tau^*\}} - \ln(\tau^*) 1_{\{ r < \tau^*\}} }{ \sqrt{\tau_d -r}} + o_P(1)\\
	&= \tau^* + o_P(1),
\end{align*}
where $\tau_d = T_d/T$.